\documentclass[12pt,UKenglish]{article}
\usepackage[a4paper, total={7in, 9in}]{geometry}
\usepackage[utf8]{inputenc}
\usepackage{graphicx}
\usepackage{relsize}
\usepackage{booktabs}
\usepackage{caption}
\usepackage{natbib}
\usepackage{apalike}
\usepackage{graphicx}
\usepackage{amsmath}
\usepackage{amssymb}
\usepackage{float}
\usepackage{appendix}
\usepackage{cases}
\usepackage{amsthm}
\usepackage{stmaryrd}
\usepackage{diagbox}
\usepackage{color}
\usepackage{mathtools}
\usepackage{enumerate}
\usepackage[unicode,colorlinks]{hyperref}
\usepackage[ruled,vlined]{algorithm2e}
\usepackage{caption}
\usepackage{subcaption}
\usepackage{setspace}
\newcommand{\comment}[1]{}
\def\einf{{\rm ess \, inf}}
\def\esup{{\rm ess \, sup}}

\newtheorem{theorem}{Theorem}[section]
\newtheorem{lem}{Lemma}[section]
\newtheorem*{theorem*}{Theorem}
\newtheorem{assumption}{Assumption}
\newtheorem*{assumption*}{Assumption}
\newtheorem{proposition}{Proposition}[section]
\newtheorem*{proposition*}{Proposition}

\newtheorem*{corollary*}{Corollary}
\newtheorem{definition}{Definition}[section]
\newtheorem*{proof*}{Proof}
\newtheorem{remark}{Remark}[section]

\begin{document} 
	\setcounter{page}{1}
	
	\title{AHEAD : \textit{Ad-Hoc} Electronic Auction Design}

	\author{
		Joffrey Derchu\footnote{\'Ecole Polytechnique, CMAP; joffrey.derchu@polytechnique.edu},~ Philippe Guillot\footnote{Autorit\'e des March\'es Financiers; p.guillot@amf-france.org},~ Thibaut Mastrolia\footnote{\'Ecole Polytechnique, CMAP; thibaut.mastrolia@polytechnique.edu}~~and Mathieu Rosenbaum\footnote{\'Ecole Polytechnique, CMAP; mathieu.rosenbaum@polytechnique.edu} 
		}

\maketitle

\begin{abstract}
We introduce a new matching design for financial transactions in an electronic market. In this mechanism, called \textit{ad-hoc} electronic auction design (AHEAD), market participants can trade between themselves at a fixed price and trigger an auction when they are no longer satisfied with this fixed price. In this context, we prove that a Nash equilibrium is obtained between market participants. Furthermore, we are able to assess quantitatively the relevance of \textit{ad-hoc} auctions and to compare them with periodic auctions and continuous limit order books. We show that from the investors' viewpoint, the microstructure of the asset is usually significantly improved when using AHEAD.
\end{abstract}

\section{Introduction}

\subsection{Existing market models: Continuous limit order book and periodic auctions}
The question of a suitable market microstructure enabling an exchange to ensure satisfactory conditions for trading activities of market participants is particularly intricate. The most standard approach, adopted by a large number of exchanges, is the continuous limit order book (CLOB for short). In this setting, market participants can either choose to trade immediately by accepting the price offered by a counterparty in the order book (sending what is called an {\it aggressive} order and thereby reducing the quantity of shares instantly available in the limit order book) or place a {\it passive} order, which waits in the order book to find a counterparty. The recent change in the very nature of market makers, which are nowadays essentially high frequency traders, has triggered a debate on whether CLOBs are the most suitable order matching mechanism, notably in terms of quality of the price formation process. The alternative design which is usually put forward is that of periodic auctions. In this case, transactions occur once the auction terminates. The traded price is the equilibrium price maximising the number of financial instruments traded, determined at the end of the auction period from the imbalance between buy and sell orders accumulated during the duration of the auction. Currently, some auctions are already held at regular intervals in many markets where the main mechanism is a CLOB, typically at the beginning and at the end of the trading day. Moreover, some exchanges organise periodic auctions throughout the day. This is for example the case of BATS-Cboe for European equities.\\

One of the benefits of auctions derives from the fact that they mechanically slow down the market. Doing so, they suppress some obvious flaws due to speed competition of high frequency market makers in a CLOB environment. This is particularly well emphasized in \cite{farmer,budish2} and the influential paper \cite{budish} where a lower bound for auction duration so that speed arbitrages vanish is provided (about 100 ms). In the paper \cite{duzhu}, the authors also consider the issue of determining a suitable time period for the auction duration. To do so, they model the behaviour of microscopic agents who optimise their demand schedules with respect to the available information in the market. They show that the optimal auction duration is linked to the rate of arrival of information.\\

Regarding a suitable market design, each type of market participants has a different view on the question depending on its activity. This is why there is a crucial need for a quantitative analysis enabling us to assess and compare the different mechanisms objectively. This is done in \cite{jusselin2019optimal} where the authors extend the works \cite{fricke2018too} and \cite{garbade}. More precisely, they are able to compare CLOBs and periodic auctions from a price formation process viewpoint using stochastic differential games. They also provide optimal auction durations (a few minutes in practice according to their approach, depending on the asset involved).

\subsection{Going AHEAD}

Auctions and CLOB represent two quite orthogonal approaches in terms of market design. In this paper, we aim to study an hybrid mechanism that we call \textit{ad-hoc} electronic auction design (AHEAD). The idea of \textit{ad-hoc} auctions is to organise a specific type of continuous trading session after each auction. During the continuous session, market participants trade between themselves at a fixed price equal to the last auction's clearing price. Any market participant has the opportunity to end the continuous session when he is no longer satisfied with the price by triggering a new auction. The only constraint imposed by the exchange to market participants for triggering a new auction phase is to commit at least a minimal volume in the auction. In this setting, there can be two reasons to motivate investors for ending the continuous phase. Either they consider the trading price is no longer reasonable or they are not able to trade at this price because of the lack of counterparty from other participants. Our underlying idea for the relevance of this mechanism is that it can provide the best of both worlds, interpolating between CLOB and periodic auctions, by conveying information about a potential price change to all market participants in a timely manner. On the one hand, auction phases enable market participants to source liquidity through a competitive process of price formation. On the other hand, potential local volume disequilibria between the needs of buyers and sellers that do not warrant a price change can be mitigated during the continuous sessions. \\

Note that we focus here on AHEAD implementation on non-fragmented markets, such as small and mid-cap markets (some of these stocks may display a limited fragmentation, however discussions towards the revision of MiFID II in Europe indicate the will of the regulator to impose a unique structure for such assets). In this case, AHEAD could be a decisive model since it improves liquidity aggregation. Furthermore, an auction on an illiquid instrument ending without any transaction would still deliver a change in the clearing price of the instrument.
Direct listings, where building steadily liquidity is the key success factor, represent another situation where AHEAD could prove worthwhile and allow easier access to the financial markets for the small and mid-cap enterprises. Competition between two AHEAD markets could result in something closer to a more stable version of a CLOB: in phases where both venues can trade on a fixed price, there would be situations where one venue would display liquidity at a ``bid" and the other venue at an ``offer" (depending on the chosen make-take fees schedule)\footnote{It would then be critical for the regulator to prevent a ``race to the bottom" between competing venues by setting minimum values for the triggering quantities and the auction durations, as MiFID II did for the tick size.}. In highly fragmented markets, CLOBs and auctions interact by catering to different strategies from market participants. It would be very complex to model such interactions if AHEAD were to be added to the current microstructure, since the market participant mix would probably differ considerably from one venue to another. Such study is left for further research.\\

We consider three agents in our model: two investors, one buyer and one seller, using aggressive orders and one market maker using passive orders. The market maker provides liquidity during both the continuous and auction phases. During the continuous phases, he accepts transactions at the last auction's clearing price provided they are profitable. To assess the profitability of a transaction, the market maker compares the last clearing price and the current efficient price, that is assumed to be observed/built by him continuously\footnote{In a further study, the model could be developed to allow the market maker to manage its inventory by triggering auctions himself and investors to use both passive and aggressive orders.}. Our buyer (resp. seller) investor wishes to buy (resp. sell) a given amount of shares over a given time period. More specifically, we consider that he aims at following a trading intensity target (coming for example from an Almgren-Chriss type algorithm, see \cite{almgren2001optimal}). Thus his goal is to optimise his PnL while staying close to the target. From a mathematical viewpoint, his objective function consists into two terms that he wants to minimise: one measuring his realized trading costs and the other the deviation from the target. To achieve their goal, our investors have access to two controls: the trading rate with which they send their market orders and the triggering times of the auctions. They optimise simultaneously and without communication their strategies. Note that there are of course more than two investors in an actual market. However, since our auction period will be quite short, we expect in practice only a small number of investors to take part in each auction (these investors being probably different from one auction to the other). Note also that, in a live market environment, participants are not restricted to aggressive orders and also compete through passive orders: in an AHEAD market, an aggressive order greater than the liquidity waiting for execution in the order book would become a passive order for the remainder of the order.\\

In this model, we show that the market admits a Nash equilibrium. This implies that \textit{ad-hoc} auctions are a viable design as a trading mechanism. Furthermore, from our theoretical results, we can build a numerical methodology enabling us to compute the optimal strategies and value functions of the investors under various market configurations. This is not only done in the \textit{ad-hoc} auction framework but also under CLOB and periodic auction markets. This allows us to provide a quantitative assessment of the AHEAD market from the investors' viewpoint and to compare it with the CLOB and periodic auction structures.

\subsection{AHEAD contribution}
Our main findings are the following. First AHEAD seems to be systematically preferable than CLOBs from a market taker perspective. This is somehow in line with the results in \cite{budish,jusselin2019optimal} which underline the relevance of auctions compared to CLOBs. Furthermore, based on our computations of the value functions, we conclude that for a large investor, \textit{ad-hoc} auctions are always a suitable design (even compared with periodic auctions), in particular when the other investor is smaller. It enables the large investor to execute part of his orders with the market maker and to launch auctions when he really needs to do so. In addition to that, thanks to the transactions executed with the market maker during the continuous phase, he reduces its volume imbalance with respect to the smaller investor during the auctions phases. For a small investor, strategic considerations play an important role in the comparison between \textit{ad-hoc} and periodic auctions. Essentially, if a small investor is still large enough to be able to trigger auctions without too much relative cost, the \textit{ad-hoc} auction mechanism is beneficial for him. Otherwise, periodic auctions are more attractive from this investor's viewpoint. In practice, in an actual market, the smaller investor could in fact even place passive orders and hence profit from the market impact generated by the larger one. Therefore a very small investor may prefer periodic auctions on instruments with high price viscosity/long queuing time because, in that case, the larger one cannot benefit from the continuous phase to reduce his volume imbalance in comparison to the smaller investor, leading to very favourable auction clearing prices for the latter.\\

The paper is organised as follows. In Section \ref{sec2} we describe the \textit{ad-hoc} auction mechanism and our model. We introduce in Section \ref{sec3} the notion of equilibrium in our framework and provide results about the existence of such equilibrium under various types of assumptions. Numerical experiments and economic insights can be found in Section \ref{sec4}. The proofs are relegated to the Appendix. 

\section{Model}
\label{sec2}
In this section, we introduce our model for a market with \textit{ad-hoc} auctions. We build our mathematical framework and explain how our market participants (the two market takers and the market maker) interact. Then we describe the objectives of those participants in terms of optimisation problems.

\subsection{Framework}
Let $T>0$ be a final horizon time, $h>0$ the auction's duration, $\Omega_c$ the set of continuous functions from $[0,T+h]$ into $\mathbb{R}$, $\Omega_d$ the set of piece-wise constant c\`adl\`ag functions from $[0,T+h]$ into $\mathbb{N}$, and $\Omega = \Omega_c\times(\Omega_d)^2$ with corresponding Borel algebra $\mathcal{F}$. The observable state is the canonical process $(W_t,\Tilde{N}^a_t,\Tilde{N}^b_t)_{ t\in[0,T+h]}$ on the measurable space $(\Omega,\mathcal{F})$ defined for any $t\in[0,T+h] \text{ and }\omega= (w,n^a,n^b)\in\Omega$ by

\begin{equation*}
    W_t(\omega) :=w(t), \Tilde{N}^a_t(\omega) :=n^a(t), \Tilde{N}^b_t(\omega) :=n^b(t),
\end{equation*}
with canonical completed filtration $\mathbb{F}= (\mathcal{F}_t)_{t\in[0,T+h]}= (\mathcal{F}^c_t\otimes(\mathcal{F}^d_t)^{\otimes 2})_{t\in[0,T+h]}$.\\ 

The trading universe is reduced to a single risky asset with observable efficient price $P^*$ given by
\begin{equation*}
    P^*_t:=P^*_0+\sigma W_t\text{, }t\in[0,T+h],
\end{equation*}
with initial price $P^*_0>0$ and constant volatility $\sigma>0$. The probability measure on $\Omega$ will be defined so that $W$ is a Brownian motion. The efficient price is to be understood as a benchmark price that market participants use to measure their trading costs by comparing it with the price they get in their actual transactions, see for example \cite{delattre2013estimating,robert2011new,stoikov2018micro}. The processes $\Tilde{N}^a$ and $\Tilde{N}^b$ will correspond to the quantities of orders sent by our two investors.

\subsection{The market takers}
We consider two investors (market takers) sending aggressive orders only. We call them Player $a$ and Player $b$. Player $a$ only sends buy market orders while Player $b$ only sends sell market orders.
Let $\lambda_->0$ be the minimum intensity of arrival of orders and $\lambda^+>\lambda_-$ the maximum intensity. We equip our filtered space with the probability $\mathbb{P}^W\otimes\mathbb{P}^N$ where $\mathbb{P}^W$ is the Wiener measure and $\mathbb{P}^N$ is the solution to the martingale problem (in the sense of \cite{jacodshiryav})
\begin{equation*}
M_t = (\Tilde{N}^a_t,\Tilde{N}^b_t)^T - t\mathcal{L}_0 \text{ with }\mathcal{L}_0 = (\lambda^0,\lambda^0)^T\text{, }0<\lambda^0<\lambda^+\text{, }t\in[0,T+h]
\end{equation*}
on $((\Omega_d)^2,\mathcal{B}((\Omega_d)^2),((\mathcal{F}^d_t)^{\otimes 2})_{t\in[0,T+h]})$.\\

In our model, Player $a$ and Player $b$ control the intensities of buy and sell orders respectively. The set of admissible controls denoted by $\mathcal{U}$ is defined by all predictable processes with values in $[\lambda_-, \lambda_+]$. For any pair $(\lambda^a,\lambda^b)$ of admissible controls, we associate $\mathbb{P}^{\lambda^a,\lambda^b}$ the measure defined by
\begin{equation*}
    \frac{d\mathbb{P}^{\lambda^a,\lambda^b}}{d\mathbb{P}}
    \bigg|_t= \Psi_t^{\lambda^a,\lambda^b},
\end{equation*}
where $\Psi_t^{\lambda^a,\lambda^b}$ is the Doleans-Dade exponential martingale given by
\begin{equation*}
    \Psi_t^{\lambda^a,\lambda^b} = \exp\Big(\int_0^t \big(\log(\frac{\lambda^a_s}{\lambda^0})d\Tilde{N}^a_s-(\lambda^a_s-\lambda^0)ds+ \log(\frac{\lambda^b_s}{\lambda^0})d\Tilde{N}^b_s-(\lambda^b_s-\lambda^0)ds\big)\Big).
\end{equation*}
Thus, under the measure $\mathbb{P}^{\lambda^a,\lambda^b}$, the processes $(\Tilde{N}^a_s-\int_0^s \lambda^a_u du)_{0\leq s\leq T+h}$, $(\Tilde{N}^b_s-\int_0^s \lambda^b_u du)_{0\leq s\leq T+h}$ are martingales and $(W_s)_{0\leq s\leq T}$ is still a Brownian motion independent of the processes $(\Tilde{N}^a, \Tilde{N}^b)$.
In the following, we denote by $\mathbb{E}^{\lambda^a,\lambda^b}$ the expectation under $\mathbb{P}^{\lambda^a,\lambda^b}$ and we write
$\Psi_{s,t}^{\lambda^a,\lambda^b} =  \Psi_t^{\lambda^a,\lambda^b}/ \Psi_s^{\lambda^a,\lambda^b}$ for $s\leq t$.\\



The market takers can trigger an auction and we focus on analysing the market and the behaviours of the participants until the end of the auction. We do not consider successive auction phases as it would lead to important additional technical difficulties. Furthermore, we may expect that in practice, under AHEAD, the market would be quite regenerative from one phase to the other. We write $\mathcal{T}_{s,t}$ with $0\leq s\leq t\leq T+h$ for the set of stopping times taking values in $[s,t]$ and denote by $\tau^a$ and $\tau^b$ in $\mathcal{T}_{0,T}$ the stopping times chosen by Player $a$ and Player $b$ respectively. An auction starts at time $\tau=\tau^a\wedge\tau^b$, considering that if no player triggers an auction before time $T$, an auction is automatically triggered at time $T$.\\

Let $(\tau,\tilde{\tau})\in \mathcal T_{0,T+h}^2$ be such that $\tau\leq \tilde{\tau}, \; \mathbb P-$a.s. and $\lambda\in \mathcal U$. We denote by $\lambda_{[\tau,\tilde{\tau}]}$ and $\mathcal{U}_{[\tau,\tilde\tau]}$ the restriction of $\lambda$, respectively $\mathcal{U}$, to $[\tau,\tilde{\tau}]$. For any $\lambda\in \mathcal U$ and $\mu\in \mathcal U_{[\tau,T+h]}$, we set $(\lambda\otimes_\tau\mu)_u := \lambda_u\mathbf{1}_{u\leq\tau}+\mu_u\mathbf{1}_{\tau<u},\; u\leq T+h.$  \\

Finally, we introduce a mechanism which forces the market taker who initiates an auction to trade a minimal amount in it. This is obviously because from an exchange or regulator viewpoint, only meaningful auctions are relevant. This means auctions should take place when the price $P$ is no longer satisfactory. Requiring a minimal traded volume tends to make the auction clearing price go against the market participant who has triggered the auction. Consequently, one triggers an auction when really needed. This can also be seen as a constraint or a cost associated with triggering an auction, where the market participant considers this cost is less than the cost of waiting with a passive order placed in the order book. 
Thus we assume that a fixed given number of orders $\hat{n}\in\mathbb{N}$ is automatically recorded by the exchange for a player triggering an auction. In case both players triggers at the same time (which will be unlikely but possible in theory in our discrete setting), we write  $\hat{n}_{ab}$ for this number. We define two $\mathcal{F}_\tau$-measurable random variables, $N^a_+$ and $N^b_+$,  representing the number of orders automatically recorded by the exchange for Player $a$ and Player $b$ when the auction starts, that is
\begin{equation*}
\begin{split}
    N^a_+=& \hat{n}\mathbf{1}_{\tau^a<\tau^b,\tau^a<T}+\hat{n}_{ab}\mathbf{1}_{\tau^a=\tau^b<T}\\
    N^b_+=& \hat{n}\mathbf{1}_{\tau^b<\tau^a,\tau^b<T}+\hat{n}_{ab}\mathbf{1}_{\tau^a=\tau^b<T}.
\end{split}
\end{equation*}
\begin{remark}
In practice, in a continuous-time market, the two players would of course never trigger an auction at the same time as the matching engine needs anyway to process one message first. It is actually a straightforward extension to consider the case where for Player $a$, $\hat{n}_{ab}$ is replaced by a random variable taking values $0$ or $\hat{n}$ with probability $0.5$ and for Player $b$ by $\hat{n}$ minus this variable. We will actually consider such situation in the numerical results of Section \ref{sec4} but keep $\hat{n}_{ab}$ for simplicity for the theoretical developments. In addition, note that we can very well think of a situation where the exchange would let participants trigger auctions only at some (frequent) specific times. 
\end{remark}

\subsection{The market makers}
\subsubsection{Continuous trading phase}
Let $P\in\mathbb{R}$ be a price fixed at $t=0$. During the continuous phase, at time $t$, the market maker accepts an order from Player $a$ (buy order) if $P>P^*_t$. In this case, a unit quantity is traded at price $P$. Symmetrically, he accepts an order from Player $b$ (sell order) if $P<P^*_t$ and then a unit quantity is traded at price $P$. In other words, at time $t$ during the continuous trading phase, Player $a$ pays $P\mathbf{1}_{P>P^*_t} d\Tilde{N}^a_t$ to buy $\mathbf{1}_{P>P^*_t}d\Tilde{N}^a_t$, while Player $b$ earns $P\mathbf{1}_{P<P^*_t} d\Tilde{N}^b_t$ from the selling of $\mathbf{1}_{P<P^*_t}d\Tilde{N}^b_t$.\\

We introduce the processes $N^a$ and $N^b$ describing the number of orders sent by Player $a$ and Player $b$ which are not rejected by the market maker. They are defined by
$$
    N^a_t = \int_0^t (\mathbf{1}_{s\leq\tau}\mathbf{1}_{P>P^*_s}+\mathbf{1}_{s>\tau})d\Tilde{N}^a_s,~~
    N^b_t = \int_0^t (\mathbf{1}_{s\leq\tau}\mathbf{1}_{P<P^*_s}+\mathbf{1}_{s>\tau})d\Tilde{N}^b_s.
$$

\subsubsection{Auction}
During the auction, the market maker is willing to buy or sell a given quantity at a certain price. We consider that he provides a mid-price, that we naturally take equal to $P^*_{\tau+h}$ and a slope $K\in\mathbb{R}$, meaning that he offers a volume $K(p-P^*_{\tau+h})$ at time ${\tau+h}$ when the auction price is $p\in\mathbb{R}$. Player $a$ sends $N^a_{\tau+h}-N^a_\tau+N^a_+$ buy market orders during the auction and Player $b$ sends $N^b_{\tau+h}-N^b_\tau+N^b_+$ sell market orders. So and similarly to \cite{jusselin2019optimal}, the auction clearing price $P^{auc}$ fixed at the clearing time $\tau+h$ is solution of the equation which equals supply and demand:
\begin{equation*}
    0=-K(P^{auc}-P^*_{\tau+h})+(N^a_{\tau+h}-N^a_\tau+N^a_+)-(N^b_{\tau+h}-N^b_\tau+N^b_+)
\end{equation*}
\textit{i.e.}
\begin{equation}\label{eq:pauc}
    P^{auc}=P^*_{\tau+h}+\frac{(N^a_{\tau+h}-N^a_\tau+N^a_+)-(N^b_{\tau+h}-N^b_\tau+N^b_+)}{K}.
\end{equation}
Thus, at the end of the auction, Player $a$ buys $N^a_{\tau+h}-N^a_\tau+N^a_+$ units at price $P^{auc}$ and Player $b$ sells $N^b_{\tau+h}-N^b_\tau+N^b_+$ units at price $P^{auc}$.

\begin{remark}\label{rem::mm_price}
One could think the market maker should rather take a mid-price equal to $\pm\infty$ if $N^a_{\tau+h}-N^a_\tau+N^a_+\lessgtr N^b_{\tau+h}-N^b_\tau+N^b_+$ to optimise his PnL. However, in a real market, market makers send limit orders over a bounded price interval and competition between them prevents them from displaying irrealistic prices. Also, there is in practice uncertainty on the traded volumes (notably because auctions durations are slightly randomised). High uncertainty would lead to a high value of $K$ to compensate the lack of information on $(N^a_{\tau+h}-N^a_\tau+N^a_+)-(N^b_{\tau+h}-N^b_\tau+N^b_+)$.
\end{remark}

\subsection{Objectives}
Both market takers wish to optimise their PnL per unit of time. We suppose that they compare the prices they get to the efficient price $P^*$ seen as a benchmark. Moreover, they aim at trading a certain number of assets per unit of time (respectively $v^a$ and $v^b$ units per second) and have to pay penalties if they do not reach those targets. We now give an explicit decomposition of their trading costs per unit of time.

\subsubsection{Costs during the continuous trading phase}
As explained above, we assume that our two players are penalised during the continuous market phase if they do not trade the right volumes. More precisely, during the continuous market phase, we consider the costs of Player $a$ and Player $b$ are respectively given for any $t\in[0,T+h]$ by
\begin{equation}\label{eq::costs}
\begin{split}
    L^a_t &=q\int_0^{t\wedge \tau} (v^a s-N^a_s)^2ds  +\int_0^{t\wedge\tau} (P-P^*_t)\mathbf 1_{P>P^*_t} dN^a_t\\
    L^b_t &= q\int_0^{t\wedge \tau} (v^b s-N^b_s)^2ds-\int_0^{t\wedge\tau} (P-P^*_t)\mathbf 1_{P<P^*_t} dN^b_t.
\end{split}
\end{equation}

The first term of these equations, where $q>0$, represents the penalty if the number of trades does not match the targeted value and the second one is the cost resulting from trading activities compared to the benchmark price $P^*$.

\begin{remark}
We could also compute the actual trading costs instead of the costs with respect to the efficient price, replacing $P-P^*_t$ by $P$ in \eqref{eq::costs}.
\end{remark}

\subsubsection{Costs during the auction}
We now turn to the costs Player $a$ and Player $b$ are subjected to during the auction. We assume again that both players are penalised during the auction if they do not trade at the rate $v^a$ and $v^b$ respectively. The penalty here is also quadratic with parameter $q>0$. Thus the penalties of Player $a$ and Player $b$  during the auction are respectively given by
\begin{equation*}
    \mathcal{C}^a_{auc} = qh(v^a(\tau+h)-N^a_{\tau+h}-N^a_+)^2
\text{ and  }
    \mathcal{C}^b_{auc} = qh(v^b(\tau+h)-N^b_{\tau+h}-N^b_+)^2.
\end{equation*}
As in \cite{jusselin2019optimal}, the cost resulting from trading activities of market taker $a$  is given by $N^a_{\tau,\tau+h}(P^{auc}-P^*_{\tau+h})$ while the gain of $b$ resulting from his trades is $N^b_{\tau,\tau+h}(P^{auc}-P^*_{\tau+h})$, where $N^a_{\tau,\tau+h} = N^a_{\tau+h}-N^a_{\tau}+N^a_+$ and $N^b_{\tau,\tau+h} = N^b_{\tau+h}-N^b_{\tau}+N^b_+$.\\

Putting together all the costs/gains of our market takers and using Equation \eqref{eq:pauc}, we get that the total cost of Player $a$ per unit of time is given by
\begin{equation*}
     \frac{L^a_\tau+\mathcal{C}^a_{auc}+N^a_{\tau,\tau+h}(P^{auc}-P^*_{\tau+h})}{\tau+h} = \frac{L^a_\tau+\mathcal{C}^a_{auc}+\frac{N^a_{\tau,\tau+h}\Delta N_{\tau,\tau+h}}{K}}{\tau+h}
\end{equation*}
while the gain of Player $b$ is
\begin{equation*}
    \frac{-L^b_\tau-\mathcal{C}^b_{auc}+N^b_{\tau,\tau+h}(P^{auc}-P^*_{\tau+h})}{\tau+h} = \frac{-L^b_\tau-\mathcal{C}^b_{auc}+\frac{N^b_{\tau,\tau+h}\Delta N_{\tau,\tau+h}}{K}}{\tau+h},
\end{equation*}
where we set $\Delta N_{\tau,\tau+h} = N^a_{\tau,\tau+h}- N^b_{\tau,\tau+h}.$
For $x_0\in\mathbb R_+\times\mathbb N\times\mathbb N\times\mathbb R\times\mathbb R$ and a pair of controls $((\tau^{a},\lambda^{a}),(\tau^{b},\lambda^{b}))$,  let 
\begin{equation}\label{obj:playera}
    J^a(x_0,(\tau^a,\lambda^a),(\tau^b,\lambda^b)) = \mathbb{E}^{\lambda^a,\lambda^b}\big[\frac{L^a_\tau+\mathcal{C}^a_{auc}+\frac{N^a_{\tau,\tau+h}\Delta N_{\tau,\tau+h}}{K}}{\tau+h}\big|(P^*_0,N^a_0,N^b_0,L^a_0,L^b_0)=x_0\big]
\end{equation}
and 
\begin{equation}\label{obj:playerb}
 J^b(x_0,(\tau^a,\lambda^a),(\tau^b,\lambda^b)) = \mathbb{E}^{\lambda^a,\lambda^b}\big[\frac{-L^b_\tau-\mathcal{C}^b_{auc}+\frac{N^b_{\tau,\tau+h}\Delta N_{\tau,\tau+h}}{K}}{\tau+h}\big|(P^*_0,N^a_0,N^b_0,L^a_0,L^b_0)=x_0\big].
\end{equation}

Since Player $a$ aims at minimising his cost, his goal is to minimise over $(\tau^a,\lambda^a)$ the objective function
\begin{equation}
    J^a(x_0,(\tau^a,\lambda^a),(\tau^b,\lambda^b))
    \end{equation}
where $(\tau^b,\lambda^b)$ are controlled by Player $b$. Symmetrically, since Player $b$ aims at maximising his gain, his goal is to maximise over $(\tau^b,\lambda^b)$ the objective function
\begin{equation}
 J^b(x_0,(\tau^a,\lambda^a),(\tau^b,\lambda^b))
\end{equation}
where $(\tau^a,\lambda^a)$ are controlled by Player $a$.

\begin{remark}
Using Lemma C.1 from \cite{jusselin2019optimal} and the fact that $\tau\leq T$ and $h>0$, we obtain that \eqref{obj:playera} and \eqref{obj:playerb} are well-defined and finite.
\end{remark}

\section{Nash equilibrium for pure and mixed stopping games}
\label{sec3}
In this section, we investigate the existence of an equilibrium in the optimisation problems of the market takers in the sense of Nash equilibrium adapted to our framework. We start by defining the notion of open-loop Nash equilibrium in the sense of \cite{carmona}. Then 
we show that restraining the set of stopping times to those taking values in a finite set allows us to build an equilibrium in the simple case $\Tilde{n}=\Tilde{n}^{ab}=0$ and in the general case by considering generalized stopping times.

\subsection{Open-loop Nash equilibrium}\label{section:defsub}

First we define the notion of open-loop Nash equilibrium.
\begin{definition}[Open-Loop Nash Equilibrum (OLNE)]\label{def:nash1}
Given $x_0\in\mathbb R_+\times\mathbb N\times\mathbb N\times\mathbb R\times\mathbb R$, we say that the pair of controls $((\tau^{a,*},\lambda^{a,*}),(\tau^{b,*},\lambda^{b,*}))$ is an open-loop Nash equilibrum of the game (OLNE for short) if
\begin{equation*}
\begin{cases}
    J^a(x_0,(\tau^{a,*},\lambda^{a,*}),(\tau^{b,*},\lambda^{b,*}))&\leq J^a(x_0,(\tau^a,\lambda^a),(\tau^{b,*},\lambda^{b,*})) \hspace{5mm}\forall (\tau^a,\lambda^a)\in\mathcal{T}_{0,T}\times\mathcal{U}\\
    J^b(x_0,(\tau^{a,*},\lambda^{a,*}),(\tau^{b,*},\lambda^{b,*}))&\geq J^b(x_0,(\tau^{a,*},\lambda^{a,*}),(\tau^b,\lambda^b))\hspace{5mm}\forall (\tau^b,\lambda^b)\in\mathcal{T}_{0,T}\times\mathcal{U}.
\end{cases}
\end{equation*}
\end{definition}

We now define a Nash equilibrium for the auction phase.
\begin{definition}[Open-loop Nash equilibrium for the $\tau-$sub-game]
Given $x\in\mathbb R_+\times\mathbb N\times\mathbb N\times\mathbb R\times\mathbb R$ and $\tau\in\mathcal{T}_{0,T}$, we say that the pair of controls $(\mu^{a,*},\mu^{b,*})$ is an open-loop Nash equilibrium for the $\tau$-sub-game if
\begin{equation*}
\begin{cases}
    \mathbb{E}^{\mu^{a,*},\mu^{b,*}}_\tau\Big[\mathcal{C}^a_{auc}+\frac{N^a_{\tau,\tau+h}\Delta N_{\tau,\tau+h}}{K}\Big]=\underset{\mu^a\in\mathcal{U}_{[\tau,T+h]}}{\inf}\; \mathbb{E}^{\mu^{a},\mu^{b,*}}_\tau\Big[\mathcal{C}^a_{auc}+\frac{N^a_{\tau,\tau+h}\Delta N_{\tau,\tau+h}}{K}\Big]\\
    \mathbb{E}^{\mu^{a,*},\mu^{b,*}}_\tau\Big[-\mathcal{C}^b_{auc}+\frac{N^b_{\tau,\tau+h}\Delta N_{\tau,\tau+h}}{K}\Big]=\underset{\mu^b\in\mathcal{U}_{[\tau,T+h]}}{\sup}\; \mathbb{E}^{\mu^{a,*},\mu^{b}}_\tau\Big[-\mathcal{C}^b_{auc}+\frac{N^b_{\tau,\tau+h}\Delta N_{\tau,\tau+h}}{K}\Big],
\end{cases}
\end{equation*}where $\mathbb E_\tau[\cdot]:= \mathbb E[\cdot| (P^*_\tau,N^a_\tau,N^b_{\tau}, L^a_\tau,L^b_\tau)=x]$.
\end{definition}

If an open-loop Nash equilibrium exists for the $\tau$-sub-game, we write
\begin{equation}\label{sub-game}
 \begin{cases}
    \xi^a_\tau = \mathbb{E}^{\mu^{a,*},\mu^{b,*}}_\tau\Big[\mathcal{C}^a_{auc}+\frac{N^a_{\tau,\tau+h}\Delta N_{\tau,\tau+h}}{K}\Big]\\
   \xi^b_\tau = \mathbb{E}^{\mu^{a,*},\mu^{b,*}}_\tau\Big[-\mathcal{C}^b_{auc}+\frac{N^b_{\tau,\tau+h}\Delta N_{\tau,\tau+h}}{K}\Big]
    \end{cases}
\end{equation}
for the payoff of the sub-game (where a given open-loop Nash equilibrium for the $\tau$-sub-game is chosen).\\ 

Similarly to the results of \cite{hamadne2014bangbang,jusselin2019optimal}, we know that there exists an open-loop Nash equilibrium for the $\tau-$sub-game \eqref{sub-game}. Thanks to a dynamic programming argument, we can show that we can start by finding optimal controls for the sub-game starting at $\tau$ and that an OLNE for the game provides an open-loop Nash equilibrium for the sub-game corresponding to the auction phase. This is stated in the following proposition.

\begin{proposition}\label{prop:DPP}
Let $x_0\in\mathbb R_+\times\mathbb N\times\mathbb N\times\mathbb R\times\mathbb R$. For any $\tau\in\mathcal{T}_{0,T}$, there exists at least one open-loop Nash equilibrium $(\mu^{a,*},\mu^{b,*})$ to the $\tau$-sub-game. Moreover, we can find two deterministic functions with polynomial growth $g^a$, $g^b$ such that $\xi^a_\tau = g^a(N^a_{\tau}-v^a\tau,N^b_{\tau}-v^b\tau,N^a_+,N^b_+)$ and $\xi^b_\tau = g^b(N^a_{\tau}-v^a\tau,N^b_{\tau}-v^b\tau,N^a_+,N^b_+)$.\\

Finally, if $((\tau^{a,*},\lambda^{a,*}),(\tau^{b,*},\lambda^{b,*}))$ is an OLNE for the general game, then the following dynamic programming principle holds

\begin{equation}\label{DPP}
      \begin{cases}  J^a(x_0,(\tau^{a,*},\lambda^{a,*}),(\tau^{b,*},\lambda^{b,*})) = \underset{\tau^a\in\mathcal{T}_{0,T},\lambda^a\in\mathcal{U}_{[0,\tau]}}{\inf}\; \mathbb{E}^{\lambda^a,\lambda^{b,*}}\big[\frac{L^a_\tau+\xi^a_\tau}{\tau+h}\big] \\
       J^b(x_0,(\tau^{a,*},\lambda^{a,*}),(\tau^{b,*},\lambda^{b,*})) = \underset{\tau^b\in\mathcal{T}_{0,T},\lambda^b\in\mathcal{U}_{[0,\tau]}}{\sup}\; \mathbb{E}^{\lambda^{a,*},\lambda^{b}}\big[\frac{-L^b_\tau+\xi^b_\tau}{\tau+h}\big] 
      \end{cases}
\end{equation}
where $\tau=\tau^{a,*}\wedge\tau^{b,*}$ and $(\lambda^{a,*}_{[\tau,T+h]},\lambda^{b,*}_{[\tau,T+h]})$ is an open-loop Nash equilibrium for the $\tau$-sub-game \eqref{sub-game} (with payoffs $\xi^a_\tau$ and $\xi^b_\tau$), recalling that $(\lambda^{a,*}_{[\tau,T+h]},\lambda^{b,*}_{[\tau,T+h]})$ is the restriction of $(\lambda^{a,*},\lambda^{b,*})$ to $[\tau,T+h]$.
\end{proposition}
\begin{proof}
See Appendix \ref{app:proofpropsub}.
\end{proof}
It will be useful to consider the functions defined on $[0,T]\times\mathbb{N}\times\mathbb{N}$ and for $l=$a or $l=$b associated to the value of the sub-game for Player $l$ when
\begin{itemize}
\item he initiates the auction alone:
    $$g^{\text{first}}_l(s,n_a,n_b) = g^l(n_a-v^as,n_b-v^bs,\hat{n}\mathbf 1_{l=a},\hat{n}\mathbf 1_{l=b}),$$
\item he does not initiate the auction:
$$g^{\text{second}}_l(s,n_a,n_b)= g^l(n_a-v^as,n_b-v^bs,\hat{n}\mathbf 1_{l=b},\hat{n}\mathbf 1_{l=a}),$$
\item he initiates it at the same time as the other player:
$$g^{\text{sim}}_l(s,n_a,n_b)= g^l(n_a-v^as,n_b-v^bs,\hat{n}_{ab},\hat{n}_{ab}),$$
\item the auction starts at $T$:
$$g^{\text{T}}_l(s,n_a,n_b)= g^l(n_a-v^as,n_b-v^bs,0,0).$$
\end{itemize}

Note that from Proposition \ref{prop:DPP}, we know that these functions have polynomial growth.

\begin{remark}
To build an OLNE for the general game, we can start by building an equilibrium for the sub-game \eqref{sub-game} during the auction and then look for a solution to the problem \eqref{DPP}.
\end{remark}

\begin{remark}
The uniqueness of an open-loop Nash equilibrium for the $\tau-$ sub-game \eqref{sub-game} played during the auction is known to be a very intricate issue, see \cite{hamadne2014bangbang,jusselin2019optimal}. However, numerical experiments seem to indicate that there is only one Nash equilibrium for the sub-game for each value of $(P^*_\tau,N^a_{\tau}-v^a\tau,N^b_{\tau}-v^b\tau,N^a_+,N^b_+)$.
\end{remark}

Extending the results of \cite{basei2016nonzerosum} and \cite{basei2019nonzerosum} to include jump processes and expectations given by non-trivial risk measures, we can prove that the existence of an OLNE can be reduced to solving a system of fully coupled integro-partial PDEs. We refer to Appendix \ref{app:verif} for more details on it. However, we do not expect to obtain the existence of an OLNE in this case in a general setting. Nevertheless, as we will see below, assuming that the players can choose their stopping time only in a set of discrete times allows us to derive the existence of a Nash equilibrium in this slightly simplified setting.\\

From now on, we focus on stopping times with values in a discrete subset of $[0,T]$.

\subsection{Discretised stopping games}\label{sec::discrete}
We look for an OLNE in the case where the stopping times can only take discrete values. Set $\delta\in\mathbb{R}^+$ such that $\frac{T}{\delta}\in\mathbb{N}$. For $k=0,...,\frac{T}{\delta}$, we consider $\mathcal{T}^d_{k\delta,T}$ the set of stopping times with values in the set $\{k\delta, (k+1)\delta,...,T\}$ almost surely. Note that $\mathcal{T}^d_{(k+1)\delta,T}$ is included in $\mathcal{T}^d_{k\delta,T}$.

\begin{remark}
The following results can be easily extended to the case where the stopping times take values in any finite discrete set.
\end{remark}

For any $l\in\llbracket0,T/\delta-1\rrbracket$ and $(\lambda_k)_{k\in\llbracket l,T/\delta-1\rrbracket}\in\mathcal{U}_{[l\delta,(l+1)\delta]}\times...\times\mathcal{U}_{[T-\delta,T]}$, we set
\begin{equation*}
    \bigotimes_{k\in\llbracket l,T/\delta-1\rrbracket}\lambda_k:=\mathbf{1}_{l\delta\leq t\leq (l+1)\delta}\lambda_l+\sum_{k\in\llbracket l+1,T/\delta-1\rrbracket}\mathbf{1}_{k\delta<t\leq (k+1)\delta}\lambda_k.
\end{equation*}

Following \cite{hamadne2014bangbang} and \cite{jusselin2019optimal}, we can show that there exists a Nash equilibrium between two fixed discrete times as formalised in the next lemma.
\begin{lem}\label{lem::disc_int}
Let $k\in\llbracket 1,T/\delta\rrbracket$, $\hat g^a_k$ and $\hat g^b_k$ be two measurable functions defined on $[0,T+h]\times\mathbb{R}\times\mathbb{R}^2\times\mathbb{N}^2$ to $\mathbb{R}$, with polynomial growth. Then, their exists $(\lambda^a_{k-1},\lambda^b_{k-1})\in\mathcal{U}_{[(k-1)\delta,k\delta]}^2$ such that
\begin{equation*}
\begin{cases}
    \mathbb{E}^{\lambda^a_{k-1},\lambda^b_{k-1}}_{(k-1)\delta}\big[\hat g^a_k(k\delta,P^*_{k\delta},L^a_{k\delta},L^b_{k\delta},N^a_{k\delta},N^b_{k\delta})\big] =\;& \underset{\lambda^a\in\mathcal{U}_{[(k-1)\delta,k\delta]}}{\einf}\;\mathbb{E}^{\lambda^a,\lambda^b_{k-1}}_{(k-1)\delta}\big[\hat g^a_k(k\delta,P^*_{k\delta},L^a_{k\delta},L^b_{k\delta},N^a_{k\delta},N^b_{k\delta})\big]\\
    \mathbb{E}^{\lambda^a_{k-1},\lambda^b_{k-1}}_{(k-1)\delta}\big[\hat g^b_k(k\delta,P^*_{k\delta},L^a_{k\delta},L^b_{k\delta},N^a_{k\delta},N^b_{k\delta})\big] =\;& \underset{\lambda^b\in\mathcal{U}_{[(k-1)\delta,k\delta]}}{\esup}\;\mathbb{E}^{\lambda^a_{k-1},\lambda^b}_{(k-1)\delta}\big[\hat g^b_k(k\delta,P^*_{k\delta},L^a_{k\delta},L^b_{k\delta},N^a_{k\delta},N^b_{k\delta})\big].
\end{cases}
\end{equation*}
\end{lem}

In the spirit of \cite{ludkovski}, we will show the following results:
\begin{itemize}
    \item In the case where the triggering cost is null or small enough so that it does not impact the strategies during the auction, see Section \ref{section:unimportant}, we can always construct a Nash equilibrium by backward induction (see Theorem \ref{thm:backward:unimportant}).
    \item In the general case, see Section \ref{section:general}, we can construct a Nash equilibrium if we extend our probability space to allow for randomised stopping times, see Theorem \ref{thm:randomised}.
    \item In both cases, the pure open-loop Nash equilibrium or randomised open-loop Nash equilibrium of the discretised game is an $\epsilon$-Nash equilibrium of the continuous game, see Section \ref{sec:epsilonnash}.
\end{itemize}
\subsubsection{Discretised game}
We now introduce the notion of discretised game and that of Nash equilibrium in this framework.
\begin{definition}[Pure Open-Loop Nash Equilibrium for the discrete game (OLNED)]\label{def:OLNED}
Let $x\in\mathbb R_+\times\mathbb N\times\mathbb N\times\mathbb R\times\mathbb R$. We say that $((\tau^{a,*},\lambda^{a,*}),(\tau^{b,*},\lambda^{b,*}))\in(\mathcal{T}^d_{0,T}\times\mathcal{U})^2$ is a pure open-loop Nash equilibrium of the discretised game (OLNED for short) if it is a solution to the game 

\begin{equation*}
\begin{cases}
    \scalebox{0.97}{$\mathbb{E}^{\lambda^{a,*},\lambda^{b,*}}\Big[\dfrac{L^a_{\tau}+\mathcal{C}^a_{auc}+ N^a_{\tau,\tau+h}\frac{\Delta N_{\tau,\tau+h}}K}{\tau+h}\Big]=\underset{\substack{\tau^a\in\mathcal{T}^d_{0,T},\\ \lambda^a\in\mathcal{U}}}{\inf}\mathbb{E}^{\lambda^a,\lambda^{b,*}}\Big[\dfrac{L^a_{\tilde\tau^{a}}+\mathcal{C}^a_{auc}+N^a_{\tilde\tau^{a},\tilde\tau^{a}+h}\frac{\Delta N
   _{\tilde\tau^a,\tilde\tau^{a}+h}}{K}}{\tilde\tau^{a}+h}\Big]$}\\
    \scalebox{0.97}{$\mathbb{E}^{\lambda^{a,*},\lambda^{b,*}}\Big[\dfrac{-L^b_{\tau}-\mathcal{C}^b_{auc}+N^b_{\tau,\tau+h}\frac{\Delta N_{\tau,\tau+h}}{K}}{\tau+h}\Big]=\underset{\substack{\tau^b\in\mathcal{T}^d_{0,T},\\ \lambda^b\in\mathcal{U}}}{\sup} \mathbb{E}^{\lambda^{a,*},\lambda^b}\Big[\dfrac{-L^b_{\tilde\tau^{b}}-\mathcal{C}^b_{auc}+N^b_{\tilde\tau^{b},\tilde\tau^{b}+h}\frac{\Delta N_{\tilde\tau^{b},\tilde\tau^{b}+h }}{K}}{\tilde\tau^{b}+h}\Big]$}
\end{cases}
\end{equation*}
a.s., with $\tilde\tau^{a}=\tau^a\wedge\tau^{b,*}, \; \tilde\tau^{b}=\tau^{a,*}\wedge\tau^b$, $\tau=\tau^{a,*}\wedge\tau^{b,*} $ and where $\mathbb E[\cdot]:= \mathbb E[\cdot| (P^*_{0},N^a_{0},N^b_{0}, L^a_{0},L^b_{0})=x]$.
\end{definition}

For sake of simplicity, a pure OLNED will be simply called an OLNED. Inspired by the literature on optimal stopping in non-zero sum games in discrete time, see among others \cite{grigorova2017optimal,riedel2017subgame}, we aim at finding OLNED in the sense of the above definition. 
\subsubsection{Particular case: $\hat n=\hat n_{ab}=0$, no cost to trigger the auction}\label{section:unimportant}
We first consider the simple case $\hat{n}=\hat{n}_{ab}=0$. In this situation, there is no cost associated with the triggering of an auction and market takers are indifferent about stopping the game first or second. We will see that here a Nash equilibrium for the discretised game can be constructed explicitly.\\ 
For $(\tau^a,\tau^b)\in (\mathcal{T}^d_{0,T})^2$, if Player $b$'s stopping time is $\tau^b$, then the value of Player $a$ is the same whether he plays $\tau^a$ or $\tau^a\wedge\tau^b$. Symmetrically, the value of Player $b$ is the same whether he plays $\tau^b$ or $\tau^a\wedge\tau^b$. So we can simply consider strategies where both players stop at the same time. We look for a stopping time $\tau^*\in\mathcal{T}^d_{0,T}$ and trading intensities $(\lambda^{a,*},\lambda^{b,*})\in\mathcal{U}^2$ such that
\[\begin{cases}
    J^a(x_0,(\tau^{*},\lambda^{a,*}),(\tau^{*},\lambda^{b,*}))&\leq J^a(x_0,(\tau^a,\lambda^a),(\tau^{*},\lambda^{b,*})) \hspace{5mm}\forall (\tau^a,\lambda^a)\in\mathcal{T}^d_{0,T}\times\mathcal{U}\\
    J^b(x_0,(\tau^{*},\lambda^{a,*}),(\tau^{*},\lambda^{b,*}))&\geq J^b(x_0,(\tau^{*},\lambda^{a,*}),(\tau^b,\lambda^b))\hspace{5mm}\forall (\tau^b,\lambda^b)\in\mathcal{T}^d_{0,T}.\times\mathcal{U}
\end{cases}\]
for $x_0\in\mathbb R_+\times\mathbb N\times\mathbb N\times\mathbb R\times\mathbb R$.\\

We build by backward induction a process $(U^a_k, U^b_k)_{k\in\llbracket0,T/\delta\rrbracket}$, adapted to the discrete filtration $(\mathcal{F}_{\delta k})_{k\in\{0,1,...,T/\delta\in\mathbb{N}\}}$ so that, for each $k\in\{0,1,...,T/\delta\in\mathbb{N}\}$, $U^a_k$ and $U^b_k$ are the values of Player $a$ and Player $b$ at time $k\delta$ when they both play an OLNED.

\paragraph{Backward induction algorithm.}
We formally construct by backward induction an OLNED.
We start by setting
\begin{equation*}
    (U^a_{T/\delta} ,U^b_{T/\delta}) = (\frac{L^a_T+g^{\text{first}}_a(T,N^a_T,N^b_T)}{T+h}, \frac{-L^b_T+g^{\text{first}}_b(T,N^a_T,N^b_T)}{T+h})
\end{equation*}
and
\begin{equation*}
    \tau^*_{\frac{T}{\delta}} = T.
\end{equation*}

 Since the players are forced to enter an auction if they have not started one before time $T$ and because $\hat{n}=\hat{n}_{ab}=0$, these values are those of the game if the players start playing at time $T$. In this case, they play a Nash equilibrium during the auction denoted by $(\hat\lambda^a,\hat\lambda^b)$ by solving \eqref{sub-game}, which we can compute with the same numerical method as in \cite{jusselin2019optimal}.\\

In the interval $(T-\delta,T]$, the players cannot trigger an auction.  From Lemma  \ref{lem::disc_int}, we find $(\lambda^{a,*}_{\frac{T}{\delta}-1},\lambda^{b,*}_{\frac{T}{\delta}-1})\in\mathcal{U}_{[T-\delta,T]}^2$ such that 
\begin{equation*}
\begin{cases}
    \mathbb{E}^{\lambda^{a,*}_{\frac{T}{\delta}-1},\lambda^{b,*}_{\frac{T}{\delta}-1}}_{T-\delta}\big[U^a_{T/\delta}\big] = \underset{\lambda^a\in\mathcal{U}_{[T-\delta,T]}}{\einf}\mathbb{E}^{\lambda^a,\lambda^{b,*}_{\frac{T}{\delta}-1}}_{T-\delta}\big[U^a_{T/\delta}\big]\\
    \mathbb{E}^{\lambda^{a,*}_{\frac{T}{\delta}-1},\lambda^{b,*}_{\frac{T}{\delta}-1}}_{T-\delta}\big[U^b_{T/\delta}\big] = \underset{\lambda^b\in\mathcal{U}_{[T-\delta,T]}}{\esup}\mathbb{E}^{\lambda^{a,*}_{\frac{T}{\delta}-1},\lambda^b}_{T-\delta}\big[U^b_{T/\delta}\big].
\end{cases}
\end{equation*}
We set $\Tilde{\lambda}^{a,*}_{\frac{T}{\delta}-1}:= \lambda^{a,*}_{\frac T \delta-1} \otimes_T \hat\lambda^a$ and $\Tilde{\lambda}^{b,*}_{\frac{T}{\delta}-1}:= \lambda^{b,*}_{\frac T \delta-1} \otimes_T \hat\lambda^b$. \\

At time $T-\delta$, both players can choose whether to trigger an auction or not. Also, they are indifferent about who actually triggers the auction. 
If one of the players triggers an auction the values become $(\frac{L^a_{T-\delta}+g^{\text{first}}_a({T-\delta},N^a_{T-\delta},N^b_{T-\delta})}{T-\delta+h}, \frac{-L^b_{T-\delta}+g^{\text{first}}_b({T-\delta},N^a_{T-\delta},N^b_{T-\delta})}{T-\delta+h})$. Otherwise, if none of the players triggers an auction, their values are $(\mathbb{E}^{\lambda^{a,*}_{\frac{T}{\delta}-1},\lambda^{b,*}_{\frac{T}{\delta}-1}}_{T-\delta}\big[U^a_{T/\delta}\big],\mathbb{E}^{\lambda^{a,*}_{\frac{T}{\delta}-1},\lambda^{b,*}_{\frac{T}{\delta}-1}}_{T-\delta}\big[U^b_{T/\delta}\big])$. So each player compares the two possible values (\textit{i.e.} the two possible mean payoffs) and triggers an auction if and only if it is beneficial to him. Consequently,
if the following condition is satisfied:
\begin{equation*}
\begin{cases}
    \mathbb{E}^{\lambda^{a,*}_{\frac{T}{\delta}-1},\lambda^{b,*}_{\frac{T}{\delta}-1}}_{T-\delta}\big[U^a_{T/\delta}\big]<\frac{L^a_{T-\delta}+g^{\text{first}}_a({T-\delta},N^a_{T-\delta},N^b_{T-\delta})}{T-\delta+h}\\
    \mathbb{E}^{\lambda^{a,*}_{\frac{T}{\delta}-1},\lambda^{b,*}_{\frac{T}{\delta}-1}}_{T-\delta}\big[U^b_{T/\delta}\big]>\frac{-L^b_{T-\delta}+g^{\text{first}}_b({T-\delta},N^a_{T-\delta},N^b_{T-\delta})}{T-\delta+h}.
    \end{cases},
\end{equation*}
then none of the players triggers an auction and we set
\begin{equation*}
    (U^a_{T/\delta-1},U^b_{T/\delta-1}) = (\mathbb{E}^{\lambda^{a,*}_{\frac{T}{\delta}-1},\lambda^{b,*}_{\frac{T}{\delta}-1}}_{T-\delta}\big[U^a_{T/\delta}\big],\mathbb{E}^{\lambda^{a,*}_{\frac{T}{\delta}-1},\lambda^{b,*}_{\frac{T}{\delta}-1}}_{T-\delta}\big[U^b_{T/\delta}\big])
\end{equation*}
and 
\begin{equation*}
    \tau^*_{\frac{T}{\delta}-1} = \tau^*_{\frac{T}{\delta}}.
\end{equation*}
Otherwise
\begin{equation*}
    (U^a_{T/\delta-1},U^b_{T/\delta-1}) = (\frac{L^a_{T-\delta}+g^{\text{first}}_a({T-\delta},N^a_{T-\delta},N^b_{T-\delta})}{T-\delta+h}, \frac{-L^b_{T-\delta}+g^{\text{first}}_b({T-\delta},N^a_{T-\delta},N^b_{T-\delta})}{T-\delta+h}),
\end{equation*}
in which case the players trigger an auction at $T-\delta$ and so
\begin{equation*}
    \tau^*_{\frac{T}{\delta}-1} = T-\delta.
\end{equation*}
Then we iterate the procedure to build $U^a$ and $U^b$ and $(\tau^*,\Tilde{\lambda}^{a,*}),(\tau^*,\Tilde{\lambda}^{b,*})$ at any discrete time: Using again backward induction, we can show that, for every $k\in\llbracket0,T/\delta\rrbracket$, there exist two functions $g^a_k$ and $g^b_k$ such that $U^a_k = \hat{g}^a_k(k\delta,P^*_{k\delta},L^a_{k\delta},L^b_{k\delta},N^a_{k\delta},N^b_{k\delta})$ and $U^b_k=\hat{g}^b_k(k\delta,P^*_{k\delta},L^a_{k\delta},L^b_{k\delta},N^a_{k\delta},N^b_{k\delta})$. It is indeed true for $k=T/\delta$. Then, using a result from \cite{jusselin2019optimal}, we have that if this property holds for some $k\in\llbracket1,T/\delta\rrbracket$, it also holds for $k-1$. This allows us to apply the previous methodology and find appropriate $(\lambda^{a,*}_{k},\lambda^{b,*}_{k})$ on each interval. As a result $(\tau^*,\Tilde{\lambda}^{a,*}),(\tau^*,\Tilde{\lambda}^{b,*})$ is an OLNED. This backward induction is summed up in Algorithm \ref{algovaluediscrete}, see Appendix \ref{algo}. 

\paragraph{Existence of an OLNED.} The following theorem formalizes this procedure.
\begin{theorem}\label{thm:backward:unimportant}
Let $U^a,U^b,\lambda^{a,*},\lambda^{b,*}$ be defined by the backward induction in Algorithm \ref{algovaluediscrete} and set
\begin{equation*}
    \tau^{*}=\delta\inf\Big\{l\in\llbracket 0,T/\delta\rrbracket\text{, }U^a_l =\frac{L^a_{l\delta}+g^{\text{first}}_a({l\delta},N^a_{l\delta},N^b_{l\delta})}{l\delta+h}\text{ or } U^b_l=\frac{-L^b_{l\delta}+g^{\text{first}}_b({l\delta},N^a_{l\delta},N^b_{l\delta})}{l\delta+h}\Big\}.
\end{equation*}
Let $\Tilde{\lambda}^{a,*}=\big(\bigotimes_{l\in\llbracket 0,T/\delta-1\rrbracket}\lambda^{a,*}_l\big)\otimes_{\tau^*}\hat{\lambda}^{a}$ and $\Tilde{\lambda}^{b,*}=\big(\bigotimes_{l\in\llbracket 0,T/\delta-1\rrbracket}\lambda^{b,*}_l\big)\otimes_{\tau^*}\hat{\lambda}^{b}$. Then, the pair of controls $((\tau^{*},\Tilde{\lambda}^{a,*}),(\tau^{*},\Tilde{\lambda}^{b,*}))$ is an OLNED. In this case, $U^a_0$ and $U^b_0$ are the values of the discretised game for Player $a$ and Player $b$ respectively, \textit{i.e.}
\begin{equation*}
    U^a_0 = \underset{\tau^a\in\mathcal{T}^d_{0,T}, \lambda^a\in\mathcal{U}}{\inf}\mathbb{E}^{\lambda^a,\Tilde{\lambda}^{b,*}}\big[\frac{L^a_\tau+\mathcal{C}^a_{auc}+\frac{N^a_{\tau,\tau+h}(N^a_{\tau,\tau+h}-N^b_{\tau,\tau+h})}{K}}{\tau+h}\big]
\end{equation*}
where $\tau = \tau^a\wedge\tau^*$ and
\begin{equation*}
    U^b_0 =\underset{\tau^b\in\mathcal{T}^d_{0,T}, \lambda^b\in\mathcal{U}}{\sup} \mathbb{E}^{\Tilde{\lambda}^{a,*},\lambda^b}\big[\frac{-L^b_\tau-\mathcal{C}^b_{auc}+\frac{N^b_{\tau,\tau+h}(N^a_{\tau,\tau+h}-N^b_{\tau,\tau+h})}{K}}{\tau+h}\big]
\end{equation*}
where $\tau = \tau^*\wedge\tau^b$.
\end{theorem}
\begin{proof}
See Appendix \ref{proof:discrete:unimportant}
\end{proof}

\begin{remark}
Note that the condition $\hat n=\hat n_{ab}=0$ is sufficient but not necessary to obtain Theorem \ref{thm:backward:unimportant}. A weaker condition is actually $g_i^{\text{first}}=g_i^{\text{second}}=g_i^{\text{sim}}$ for $i=a,b$.
\end{remark}

Note that at each time $k\delta$, $k\in\llbracket0,\frac{T}{\delta}-1\rrbracket$, the players deal with a $2\times2$ game in which they decide whether or not they trigger an auction. The values of this game are represented in Table \ref{tablevalues}.
\begin{table}[ht]
\centering
    \begin{tabular}{||c | c c||} 
         \hline
         \backslashbox{a}{b}
        &\makebox[2em]{stops}&\makebox[2em]{continues}\\
        \hline\hline
         stops & $(\frac{L^a_{k\delta}+g^{\text{first}}_a}{k\delta+h},\frac{-L^b_{k\delta}+g^{\text{first}}_b}{k\delta+h})$ & $(\frac{L^a_{k\delta}+g^{\text{first}}_a}{k\delta+h},\frac{-L^b_{k\delta}+g^{\text{first}}_b}{k\delta+h})$  \\ [0.9ex] 
         \hline
         continues & $(\frac{L^a_{k\delta}+g^{\text{first}}_a}{k\delta+h},\frac{-L^b_{k\delta}+g^{\text{first}}_b}{k\delta+h})$ & $(\mathbb{E}^{\lambda^{a,*}_k,\lambda^{b,*}_k}_{k\delta}\big[U^a_{k+1}\big],\mathbb{E}^{\lambda^{a,*}_k,\lambda^{b,*}_k}_{k\delta}\big[U^b_{k+1}\big])$ \\ 
         \hline
    \end{tabular}
    \caption{Cost/gain for Player $a$/$b$ depending on whether Player $a$/$b$ stops or not the  discrete game played at time $k\delta$.}\label{tablevalues}
    \end{table}
		
The choice dictated by the algorithm implies a Nash equilibrium for each of those $2\times2$ games. In this setting, the existence of a Nash equilibrium in the general case (without imposing $\hat{n}=\hat{n}_{ab}=0$) is much more intricate to get. However, we can obtain such result if we consider randomised strategies. This leads us to the notion of randomised discrete stopping times as explained below.

\subsubsection{General case: randomised discrete stopping times}\label{section:general}
We now consider the general case. The procedure used to build a Nash equilibrium in Section \ref{section:unimportant} can be adapted to construct a Nash equilibrium in the general case. This can be done if we look for generalized stopping times instead of classical stopping times. We refer to \cite{coquet2007,solan2012random,touzivieil} for various optimal stopping problems dealing with this kind of stopping times.

\paragraph{Informal derivation of a mixed Nash equilibrium.} Right after $T-\delta$ and until $T$, the situation is the same as in the case where $\hat n=\hat n_{ab} = 0$. The players cannot trigger an auction so they play a Nash equilibrium (with no stopping allowed) until $T$. In that case, their values right after $T-\delta$ are $(\mathbb{E}^{\lambda^{a,*},\lambda^{b,*}}_{T-\delta}\big[U^a_{T/\delta}\big],\mathbb{E}^{\lambda^{a,*},\lambda^{b,*}}_{T-\delta}\big[U^b_{T/\delta}\big])$.
At time $T-\delta$, they are allowed to trigger an auction. The payoffs depend now on which player triggers an auction.
\begin{itemize}
    \item If Player $a$ triggers an auction and Player $b$ does not, the values are $$(\frac{L^a_{T-\delta}+g^{\text{first}}_a({T-\delta},N^a_{T-\delta},N^b_{T-\delta})}{T-\delta+h}, \frac{-L^b_{T-\delta}+g^{\text{second}}_b({T-\delta},N^a_{T-\delta},N^b_{T-\delta})}{T-\delta+h}).$$
    \item If Player $b$ triggers an auction and Player $a$ does not, the values are $$(\frac{L^a_{T-\delta}+g^{\text{second}}_a({T-\delta},N^a_{T-\delta},N^b_{T-\delta})}{T-\delta+h}, \frac{-L^b_{T-\delta}+g^{\text{first}}_b({T-\delta},N^a_{T-\delta},N^b_{T-\delta})}{T-\delta+h}).$$
    \item If both players trigger an auction, the values are $$(\frac{L^a_{T-\delta}+g^{\text{sim}}_a({T-\delta},N^a_{T-\delta},N^b_{T-\delta})}{T-\delta+h}, \frac{-L^b_{T-\delta}+g^{\text{sim}}_b({T-\delta},N^a_{T-\delta},N^b_{T-\delta})}{T-\delta+h}).$$
    \item Finally, if none of the players trigger an auction the values are $$(\mathbb{E}^{\lambda^{a,*},\lambda^{b,*}}_{T-\delta}\big[U^a_{T/\delta}\big],\mathbb{E}^{\lambda^{a,*},\lambda^{b,*}}_{T-\delta}\big[U^b_{T/\delta}\big]).$$
\end{itemize}   
Contrary to the previous case, there is some advantage to gain when the other player triggers an auction. Let $p^i_{\frac{T}{\delta}-1}=1$ if Player $i=a,b$ triggers an auction at $T-\delta$ and $0$ otherwise. 
For $p^b_{\frac{T}{\delta}-1}$ fixed, $p^a_{\frac{T}{\delta}-1}$ must be a minimiser of
\begin{equation*}
    \begin{split}
       p\in \{0,1\}\longmapsto  & \hspace{0.5em} p\; p^b_{\frac{T}{\delta}-1} \frac{L^a_{T-\delta}+g^{\text{sim}}_a}{T-\delta+h} + p(1-p^b_{\frac{T}{\delta}-1})\frac{L^a_{T-\delta}+g^{\text{first}}_a}{T-\delta+h} + (1-p)p^b_{\frac{T}{\delta}-1}\frac{L^a_{T-\delta}+g^{\text{second}}_a}{T-\delta+h} \\
        &\hspace{0.5em}+ (1-p)(1-p^b_{\frac{T}{\delta}-1})\mathbb{E}^{\lambda^{a,*}_k,\lambda^{b,*}_{\frac{T}{\delta}-1}}_{T-\delta}\big[U^a_{T/\delta}\big],
    \end{split}
\end{equation*}
 while for $p^a_{\frac{T}{\delta}-1}$ fixed, $p^b_{\frac{T}{\delta}-1}$ must be a maximiser of
\begin{equation*}
    \begin{split}
        p\in \{0,1\}\longmapsto  & \hspace{0.5em} p^a_{\frac{T}{\delta}-1} p \frac{-L^b_{T-\delta}+g^{\text{sim}}_b}{T-\delta+h} + p^a_{\frac{T}{\delta}-1}(1-p)\frac{-L^b_{T-\delta}+g^{\text{second}}_b}{T-\delta+h} + (1-p^a_{\frac{T}{\delta}-1})p\frac{-L^b_{T-\delta}+g^{\text{first}}_b}{T-\delta+h} \\
        &\hspace{0.5em}+ (1-p^a_{\frac{T}{\delta}-1})(1-p)\mathbb{E}^{\lambda^{a,*}_k,\lambda^{b,*}_{\frac{T}{\delta}-1}}_{T-\delta}\big[U^b_{T/\delta}\big].
    \end{split}
\end{equation*}
Such optimisers might not always exist or might not be unique (in the sense that we would have to decide who triggers the auction). However, both players can always find probabilities of stopping $p^a$ and $p^b$ in $[0,1]$ such that, if Player $b$ triggers an auction with probability $p^b$, the optimal probability of stopping for Player $a$ is $p^a$, and conversely, if Player $a$ triggers an auction with probability $p^a$, the optimal probability of stopping for Player $b$ is $p^b$. Additionally, it is often more natural to consider probabilities of stopping, in particular in the frequent case where both $(p^a,p^b)=(1,0)$ and $(p^a,p^b)=(0,1)$ are possible pure Nash equilibria. This describes a \textit{mixed} Nash equilibrium and simply corresponds to a solution of the convexification of the above problem.\\

There are multiple equivalent notions of random times which stop according to some probability. We use here the notion of \textit{mixed stopping times} of \cite{larakisolan, larakisolan2} to build our probability space. 
\begin{definition}[Generalized stopping time]
A generalized stopping time is a measurable function $\mu:\Omega\times[0,1]\rightarrow [0, T]$ such that for $\Lambda$-almost every $r\in[0,1]$, where $\Lambda$ denotes the Lebesgue measure, the function $\omega\rightarrow\mu(\omega, r)$ is a stopping time, \textit{i.e.} $\mu(., r)\in\mathcal{T}_{0,T}$.
\end{definition}

Our probability space then becomes $(\Omega\times[0,1]\times[0,1],\mathbb{P}\otimes\Lambda\otimes\Lambda)$ where the first extension characterizes the randomiser of Player $a$'s stopping time and the second one that of Player $b$'s stopping time. Let $0\leq s\leq t\leq T$. We denote by $\mathcal{T}^{*}_{s,t}$  the set of generalized stopping times with values in $[s,t]$. If $s/\delta\in\mathbb{N}$ and $t/\delta\in\mathbb{N}$, we also denote by $\mathcal{T}^{*,d}_{s,t}$ the set of generalized stopping times with values in $\llbracket s,t\rrbracket$.\\

We also extend the definition of Nash equilibrium in this framework. 
\begin{definition}[Mixed OLNE and mixed OLNED]
Let $x\in\mathbb R_+\times\mathbb N\times\mathbb N\times\mathbb R\times\mathbb R$. We say that $((\tau^{a,*},\lambda^{a,*}),(\tau^{b,*},\lambda^{b,*}))\in(\mathcal{T}^{*}_{0,T}\times\mathcal{U})\times(\mathcal{T}^{*}_{0,T}\times\mathcal{U}) $, resp. $(\mathcal{T}^{*,d}_{0,T}\times\mathcal{U})\times(\mathcal{T}^{*,d}_{0,T}\times\mathcal{U}) $, is a mixed OLNE, resp. mixed OLNED, if it is a solution to the game 
\begin{equation*}
\begin{cases}
    \scalebox{0.97}{$\mathbb{E}^{\lambda^{a,*},\lambda^{b,*}}\Big[\dfrac{L^a_{\tilde\tau^{a}}+\mathcal{C}^a_{auc}+\frac{N^a_{\tilde\tau^{a},\tilde\tau^{a}+h}\Delta N_{\tilde\tau^{a},\tilde\tau^{a}+h}}{K}}{\tilde\tau^{a}+h}\Big]=\underset{\substack{\tau^a\in\mathcal{T}^{*}_{0,T},\\ \lambda^a\in\mathcal{U}}}{\inf}\mathbb{E}^{\lambda^a,\lambda^{b,*}}\Big[\dfrac{L^a_{\tilde\tau^{a}}+\mathcal{C}^a_{auc}+\frac{N^a_{\tilde\tau^{a},\tilde\tau^{a}+h}\Delta N_{\tilde\tau^{a},\tilde\tau^{a}+h}}{K}}{\tilde\tau^{a}+h}\Big]$}\\
    \scalebox{0.97}{$\mathbb{E}^{\lambda^{a,*},\lambda^{b,*}}\Big[\dfrac{-L^b_{\tilde\tau^{b}}-\mathcal{C}^b_{auc}+\frac{N^b_{\tilde\tau^{b},\tilde\tau^{b}+h}\Delta N_{\tilde\tau^{b},\tilde\tau^{b}+h}}{K}}{\tilde\tau^{b}+h}\Big]=\underset{\substack{\tau^b\in\mathcal{T}^{*}_{0,T},\\ \lambda^b\in\mathcal{U}}}{\sup} \mathbb{E}^{\lambda^{a,*},\lambda^b}\Big[\dfrac{-L^b_{\tilde\tau^{b}}-\mathcal{C}^b_{auc}+\frac{N^b_{\tilde\tau^{b},\tilde\tau^{b}+h}\Delta N_{\tilde\tau^{b},\tilde\tau^{b}+h}}{K}}{\tilde\tau^{b}+h}\Big]$}
\end{cases}
\end{equation*}
resp.
\begin{equation*}
\begin{cases}
    \scalebox{0.97}{$\mathbb{E}^{\lambda^{a,*},\lambda^{b,*}}\Big[\dfrac{L^a_{\tilde\tau^{a}}+\mathcal{C}^a_{auc}+\frac{N^a_{\tilde\tau^{a},\tilde\tau^{a}+h}\Delta N_{\tilde\tau^{a},\tilde\tau^{a}+h}}{K}}{\tilde\tau^{a}+h}\Big]=\underset{\substack{\tau^a\in\mathcal{T}^{*,d}_{0,T},\\ \lambda^a\in\mathcal{U}}}{\inf}\mathbb{E}^{\lambda^a,\lambda^{b,*}}\Big[\dfrac{L^a_{\tilde\tau^{a}}+\mathcal{C}^a_{auc}+\frac{N^a_{\tilde\tau^{a},\tilde\tau^{a}+h}\Delta N_{\tilde\tau^{a},\tilde\tau^{a}+h}}{K}}{\tilde\tau^{a}+h}\Big]$}\\
    \scalebox{0.97}{$\mathbb{E}^{\lambda^{a,*},\lambda^{b,*}}\Big[\dfrac{-L^b_{\tilde\tau^{b}}-\mathcal{C}^b_{auc}+\frac{N^b_{\tilde\tau^{b},\tilde\tau^{b}+h}\Delta N_{\tilde\tau^{b},\tilde\tau^{b}+h}}{K}}{\tilde\tau^{b}+h}\Big]=\underset{\substack{\tau^b\in\mathcal{T}^{*,d}_{0,T},\\ \lambda^b\in\mathcal{U}}}{\sup} \mathbb{E}^{\lambda^{a,*},\lambda^b}\Big[\dfrac{-L^b_{\tilde\tau^{b}}-\mathcal{C}^b_{auc}+\frac{N^b_{\tilde\tau^{b},\tilde\tau^{b}+h}\Delta N_{\tilde\tau^{b},\tilde\tau^{b}+h}}{K}}{\tilde\tau^{b}+h}\Big]$}
\end{cases}
\end{equation*}
with $\tilde\tau^{a}=\tau^a\wedge\tau^{b,*}, \; \tilde\tau^{b}=\tau^{a,*}\wedge\tau^b$, $\tau=\tau^{a,*}\wedge\tau^{b,*} $ and where $\mathbb E[\cdot]:= \mathbb E^{\mathbb{P}\otimes\Lambda\otimes\Lambda}[\cdot| (P^*_{0},N^a_{0},N^b_{0}, L^a_{0},L^b_{0})=x]$.
\end{definition}

It is known (see for example \cite{shmayasolan,solan2012random,touzivieil}) that our notion of generalized stopping time is equivalent to the notion described in the informal derivation of a mixed Nash equilibrium above, where, at time $t$, each player stops with some probability based on the information $\mathcal{F}_t$. In particular, we can build a mixed OLNED using the same algorithm as in Section \ref{section:unimportant}, see Algorithm \ref{algovaluediscretegeneral} in Appendix \ref{algo}.

\paragraph{Existence of a (mixed) OLNED.} 
More formally, the following theorem based on the backward induction above provides the existence of a mixed OLNED. 
\begin{theorem}\label{thm:randomised}
Let 
\begin{equation*}
\begin{cases}
    \tau^{a,*}(.,r)=\delta\; \inf\big\{l\in\llbracket 0,T/\delta\rrbracket\text{, }1-\prod_{k=0}^l (1-p^{a}_l)\geq r)\big\}\\
    \tau^{b,*}(.,r)=\delta\; \inf\big\{l\in\llbracket 0,T/\delta\rrbracket\text{, }1-\prod_{k=0}^l (1-p^{b}_l)\geq r)\big\}
\end{cases}
\end{equation*}
where $p^a$ and $p^b$ are the discrete $\mathcal{F}_t$-adapted processes given by Algorithm \ref{algovaluediscretegeneral} in Appendix \ref{algo}.
Let $\Tilde{\lambda}^{a,*}=\bigotimes(\lambda^{a,*}_l)_{l\in\llbracket 0,T/\delta-1\rrbracket}\otimes_{\tau^*}\hat\lambda^{a}$ and $\Tilde{\lambda}^{b,*}=\bigotimes(\lambda^{b,*}_l)_{l\in\llbracket 0,T/\delta-1\rrbracket}\otimes_{\tau^*}\hat\lambda^{b}$, where the quantities on the r.h.s of the equalities are also given in Algorithm \ref{algovaluediscretegeneral} in Appendix \ref{algo}. Then the strategies $((\tau^{a,*},\Tilde{\lambda}^{a,*}),(\tau^{b,*}_k,\Tilde{\lambda}^{b,*}))$ describe a mixed OLNED.

\end{theorem}
\begin{proof}
According to \cite{shmayasolan} and \cite{solan2012random}, optimising over the set of generalized stopping times is equivalent to optimising over the set of adapted processes $p^a$ and $p^b$ describing the probability to stop at each discrete time. Then Theorem 1 in \cite{shmayasolan} gives a way to build the generalized stopping times from the probability processes. The rest of the proof is similar to that in the pure case and thus follows the proof of Theorem \ref{thm:backward:unimportant}.
The only difference is that the players must play the game of Table \ref{tab::mixed} when choosing whether to stop at $k\delta$ or continue playing until $(k+1)\delta$. 
\begin{table}[ht]
\centering
    \begin{tabular}{||c | c c||} 
         \hline
         \backslashbox{a}{b}
        &\makebox[2em]{stops}&\makebox[2em]{continues}\\
        \hline\hline
         stops & $(\frac{L^a_{k\delta}+g^{\text{sim}}_a}{k\delta+h},\frac{-L^b_{k\delta}+g^{\text{sim}}_b}{k\delta+h})$ & $(\frac{L^a_{k\delta}+g^{\text{first}}_a}{k\delta+h},\frac{-L^b_{k\delta}+g^{\text{second}}_b}{k\delta+h})$  \\ [0.9ex] 
         \hline
         continues & $(\frac{L^a_{k\delta}+g^{\text{second}}_a}{k\delta+h},\frac{-L^b_{k\delta}+g^{\text{first}}_b}{k\delta+h})$ & $(\mathbb{E}^{\lambda^{a,*}_k,\lambda^{b,*}_k}_{k\delta}\big[U^a_{k+1}\big],\mathbb{E}^{\lambda^{a,*}_k,\lambda^{b,*}_k}_{k\delta}\big[U^b_{k+1}\big])$ \\ 
         \hline
    \end{tabular}
    \caption{Cost/gain for Player $a$/$b$ depending on whether Player $a$/$b$ stops or not the  discrete game played at time $k\delta$.}\label{tab::mixed}
    \end{table}\vspace{0.5em}
    
    Thus, at time $k\delta$, $p^a_k$ and $p^b_k$ are defined as solutions of the following linear optimisation problems. For $p^b_k$ fixed, Player $a$ chooses
\begin{equation*}
    \begin{split}
        p^a_k \in\;&\underset{p\in[0,1]}{\arg \inf}\,\Big\{ p\, p^b_k \frac{L^a_{k\delta}+g^{\text{sim}}_a}{k\delta+h} + p(1-p^b_k)\frac{L^a_{k\delta}+g^{\text{first}}_a}{k\delta+h} + (1-p)p^b_k\frac{L^a_{k\delta}+g^{\text{second}}_a}{k\delta+h} \\
        &\hspace{3em}+ (1-p)(1-p^b_k)\mathbb{E}^{\lambda^{a,*}_k,\lambda^{b,*}_k}_{k\delta}\big[U^a_{k+1}\big] \Big\},
    \end{split}
\end{equation*}
while for $p^a_k$ fixed, Player $b$ chooses
\begin{equation*}
    \begin{split}
        p^b_k \in\; &\underset{p\in[0,1]}{\arg \sup}\, \Big\{p^a_k p \frac{-L^b_{k\delta}+g^{\text{sim}}_b}{k\delta+h} + p^a_k(1-p)\frac{-L^b_{k\delta}+g^{\text{second}}_b}{k\delta+h} + (1-p^a_k)p\frac{-L^b_{k\delta}+g^{\text{first}}_b}{k\delta+h} \\
        &\hspace{3em}+ (1-p^a_k)(1-p)\mathbb{E}^{\lambda^{a,*}_k,\lambda^{b,*}_k}_{k\delta}\big[U^b_{k+1}\big]\big\}.
    \end{split}
\end{equation*}

From classical results, see for instance \cite{NeumannMorg,Nash48}, we know that the two problems above can be solved  simultaneously\footnote{It would no longer be the case in general with $p\in\{0,1\}$, \textit{i.e.} with pure stopping times, although it works if $\hat n = \hat n_{ab} = 0$ as we have seen before.}. Solving for a mixed equilibrium yields the following result:

\begin{equation}
\label{probas}
    \begin{cases}
        p^a_k =& \dfrac{\frac{L^a_{k\delta}+g^{\text{first}}_a}{k\delta+h}-\mathbb{E}^{\lambda^{a,*}_k,\lambda^{b,*}_k}_{k\delta}\big[U^a_{k+1}\big]}{-\frac{L^a_{k\delta}+g^{\text{sim}}_a}{k\delta+h}+\frac{L^a_{k\delta}+g^{\text{first}}_a}{k\delta+h}+\frac{L^a_{k\delta}+g^{\text{second}}_a}{k\delta+h}-\mathbb{E}^{\lambda^{a,*}_k,\lambda^{b,*}_k}_{k\delta}\big[U^a_{k+1}\big]}\\ 
        &\\
        p^b_k =& \dfrac{\frac{-L^b_{k\delta}+g^{\text{first}}_b}{k\delta+h}-\mathbb{E}^{\lambda^{a,*}_k,\lambda^{b,*}_k}_{k\delta}\big[U^b_{k+1}\big]}{-\frac{-L^b_{k\delta}+g^{\text{sim}}_b}{k\delta+h}+\frac{-L^b_{k\delta}+g^{\text{first}}_b}{k\delta+h}+\frac{-L^b_{k\delta}+g^{\text{second}}_b}{k\delta+h}-\mathbb{E}^{\lambda^{a,*}_k,\lambda^{b,*}_k}_{k\delta}\big[U^b_{k+1}\big]}
    \end{cases}
\end{equation}
and one can easily verify that the denominators are non-zero and that these values are in $[0,1]$ when there is no pure Nash equilibrium.
\end{proof}

\begin{remark}
The notion of probability of stopping is quite convenient for numerical computations as we can compute the value functions and the strategies by dynamic programming.
\end{remark}

\begin{subsubsection}{Existence of $\varepsilon$-OLNE}\label{sec:epsilonnash}
In this part, we explain that the previously introduced OLNEDs (see Definition \ref{def:OLNED}) provide good approximations for OLNEs (see Definition \ref{def:nash1}), in the sense of $\varepsilon$-Nash equilibria.\\

For technical reasons we need to slightly modify the definitions of $\mathcal C^a_{auc}$ and $C^b_{auc}$ replacing $v^a\tau$ and $v^b\tau$ by $\lceil v^a\tau-\frac{1}{2}\rceil$ and $\lceil v^b\tau-\frac{1}{2}\rceil$. All the previous results could be proved in this slightly modified setting. This assumption is crucial in this section as it enables us to have that when $\frac{1}{2v^a\delta}\in\mathbb{N}$ and $\frac{1}{2v^b\delta}\in\mathbb{N}$, the functions $g^{i,\text{first}}$, $g^{i,\text{sim}}$, for $i=a,b$ take the same values if we replace $t$ by $\lceil \frac{t}{\delta}\rceil\delta$. This will be a key element in the proof of the next theorems.

\begin{theorem}[OLNED and $\varepsilon-$OLNE]\label{thm:purediscrete}
Let $\delta>0$ and $((\tau^a,\lambda^a),(\tau^b,\lambda^b))\in (\mathcal{T}^d_{0,T}\times\mathcal{U})\times(\mathcal{T}^d_{0,T}\times\mathcal{U})$ be the strategies associated to a pure OLNED starting at $0$ with time-step $\delta$. Let $\varepsilon>0$. Then, for $\delta$ small enough such that $\frac{1}{2v^a\delta}\in\mathbb{N}$ and $\frac{1}{2v^b\delta}\in \mathbb N$,
\begin{equation*}
    \begin{split}
        &\mathbb{E}^{\lambda^a,\lambda^b} \big[\frac{L^a_{\tau^a\wedge \tau^b}+\xi^a_{\tau^a\wedge\tau^b}}{{\tau^a\wedge \tau^b}+h}\big]\leq \underset{\tau\in\mathcal{T}_{0,T},\lambda\in \mathcal{U}}{\inf} \mathbb{E}^{\lambda,\lambda^b} \big[ \frac{L^a_{\tau\wedge \tau^b}+\xi^a_{\tau\wedge\tau^b}}{{\tau\wedge \tau^b}+h}\big]+\varepsilon+\varepsilon^a\\
        &\mathbb{E}^{\lambda^a,\lambda^b} \big[\frac{-L^b_{\tau^a\wedge \tau^b}+\xi^b_{\tau^a\wedge\tau^b}}{{\tau^a\wedge \tau^b}+h}\big]\geq \underset{\tau\in\mathcal{T}_{0,T},\lambda\in \mathcal{U}}{\sup} \mathbb{E}^{\lambda^a,\lambda} \big[\frac{-L^b_{\tau^a\wedge \tau}+\xi^b_{\tau^a\wedge\tau}}{{\tau^a\wedge \tau}+h}\big]-\varepsilon-\varepsilon^b\\
    \end{split}
\end{equation*}
where 
\begin{align*}
    \varepsilon^a = \frac{1}{h}\underset{\lambda^a_0\in\mathcal{U},\lambda^b_0\in\mathcal{U}}{\sup}\mathbb{E}^{\lambda^a_0,\lambda^b_0} \big[\underset{[0,T]}{\sup}\max\big(\max(g_a^{\text{sim}}(t,N^a_t,N^b_t),g_a^{\text{second}}(t,N^a_t,N^b_t))-g_a^{first}(t,N^a_t,N^b_t),0\big)\big]
\end{align*}
and 
\begin{align*}
    \varepsilon^b = \frac{1}{h}\underset{\lambda^a_0\in\mathcal{U},\lambda^b_0\in\mathcal{U}}{\sup}\mathbb{E}^{\lambda^a_0,\lambda^b_0} \big[\underset{[0,T]}{\sup}\max\big(g_b^{first}(t,N^a_t,N^b_t)-\min(g_b^{\text{sim}}(t,N^a_t,N^b_t),g_b^{\text{second}}(t,N^a_t,N^b_t)),0\big)\big].
\end{align*}
\end{theorem}
\begin{proof}
See Appendix \ref{app:purediscrete}.
\end{proof}

Also, this theorem extends easily to the case of mixed OLNEs and mixed OLNEDs.
\begin{theorem}[Mixed OLNED and $\varepsilon-$mixed OLNE]
Let $\delta>0$ and $((\tau^{a}_d,\lambda^a),(\tau^{b}_d,\lambda^b))\in (\mathcal{T}^{*,d}_{0,T}\times\mathcal{U})\times(\mathcal{T}^{*,d}_{0,T}\times\mathcal{U})$ be the strategies of a mixed OLNED starting at $0$. Let $\varepsilon>0$. Then, for $\delta$ small enough such that $\frac{1}{2v^a\delta}\in\mathbb{N}$ and $\frac{1}{2v^b\delta}\in\mathbb{N}$,
\begin{equation}
    \begin{split}
        &\mathbb{E}^{\lambda^a,\lambda^b} \big[\frac{L^a_{\tau^a\wedge \tau^b}+\xi^a_{\tau^a\wedge\tau^b}}{{\tau^a\wedge \tau^b}+h}\big]\leq \underset{\tau\in\mathcal{T}^{*}_{0,T},\lambda\in \mathcal{U}}{\inf} \mathbb{E}^{\lambda,\lambda^b} \big[ \frac{L^a_{\tau\wedge \tau^b}+\xi^a_{\tau\wedge\tau^b}}{{\tau\wedge \tau^b}+h}\big]+\varepsilon+\varepsilon^a\\
        &\mathbb{E}^{\lambda^a,\lambda^b} \big[\frac{-L^b_{\tau^a\wedge \tau^b}+\xi^b_{\tau^a\wedge\tau^b}}{{\tau^a\wedge \tau^b}+h}\big]\geq \underset{\tau\in\mathcal{T}^{*}_{0,T},\lambda\in \mathcal{U}}{\sup} \mathbb{E}^{\lambda^a,\lambda} \big[\frac{-L^b_{\tau^a\wedge \tau}+\xi^b_{\tau^a\wedge\tau}}{{\tau^a\wedge \tau}+h}\big]-\varepsilon-\varepsilon^b.
    \end{split}
\end{equation}
\end{theorem}
\begin{proof}
The proof is the same as the proof of Theorem \ref{thm:purediscrete}.
\end{proof}

In practice and in our numerical experiments, the constants $\varepsilon^a$ and $\varepsilon^b$ are negligible and very often zero. This is because they are non-zero only when there is an advantage in triggering the auction right before the other player.
This is typically not the case, unless the player triggering the auction benefits from $\hat n$ to execute a large target volume
that he could not fully do within the auction because of its bounded intensity.

\begin{remark}
The condition on $\delta$ is only technical and ensures that the changes in the targets happen on the same grid as the optimisation (with mesh $\delta$). This condition would actually not be required if the targets followed, for example, Poisson processes.
\end{remark}

\end{subsubsection}

\section{Numerical results and assessment of {\it ad-hoc} auctions}
\label{sec4}

In this section, we provide numerical results enabling us to draw conclusions on the relevance of {\it ad-hoc} auctions compared to CLOB and periodic auctions. We also discuss some implementation details. The value functions shown are multiplied by $10^6$ for more readability.

\subsection{Sub-game values depending on players' positions at the auction triggering}
First we show how the value of the sub-game played during the auction phase varies with the parameters, for Player $a$ and Player $b$. From now on we take $K = 10, v^a=v^b=0.1$. 

\subsubsection{Effect of the amount traded before the auction}\label{inventiry:graph}
We fix $h=30s$ and plot the value of the sub-game $\xi^a_\tau$ as a function of $N^a_\tau-v^a\tau$, for various values of $N^a_+$, with $N^b_+=0$ and $N^b_\tau-v^b\tau=0$ (see Figure \ref{xiaer}, left side). First we notice that these graphs are increasing with respect to $N^a_+$, which is obviously not surprising. The effect of $N^a_+$ gets more important as $N^a_\tau-v^a\tau$ becomes larger. This is because in such situation, Player $a$ is already in advance regarding to his target. A large $N^a_+$ implies that he is even more in advance and gets penalised via the objective function.\\

Looking at the graphs for fixed $N^a_+$, we see that the best context to trigger an auction is when $N^a_t-v^at$ is close to zero and actually slightly negative. In that case, Player $a$ can launch an auction without overshooting his target at the end because of the mandatory volume $N^a_+$. Moreover, we note that $\xi^a_\tau$ is large when either $N^a_\tau-v^a\tau$ is large, since Player $a$ is penalised for overshooting his target, or when $N^a_\tau-v^a\tau$ is too small, since Player $a$ has to send a lot of orders during the auction, which makes the price increase.\\

Next we plot $\xi^a_\tau$ as a function of $N^b_\tau-v^b\tau$, for various values of $N^a_+$, with $N^b_+=0$ and $N^a_\tau-v^a\tau=0$ (see Figure \ref{xiaer}, right side). We notice that $\xi^a_\tau$ is always increasing with respect to $N^b_\tau-v^b\tau$: the more Player $b$ trades before the auction in comparison to his target, the less he trades during the auction, and the higher the final price of the auction is. In addition to that, $\xi^a_\tau$ converges when $N^b_\tau-v^b\tau\rightarrow\pm\infty$: if $N^b_\tau-v^b\tau$ is too large, Player $b$ stops trading completely and if $N^b_\tau-v^b\tau$ is too small, Player $b$ would rather pay some penalties than send too many orders during the auction leading to a bad price since Player $a$ is at equilibrium when the auction starts.

\begin{figure}[ht]
\centering
\includegraphics[width=15cm]{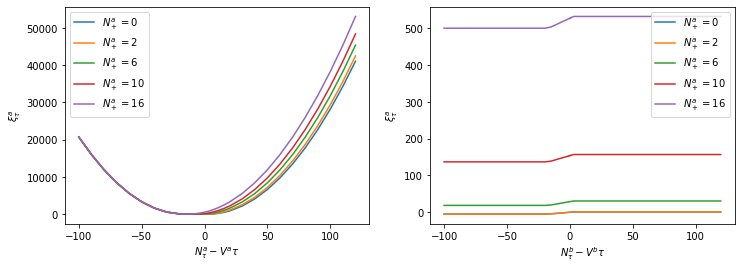}
\caption{On the left, $\xi^a_\tau$ as a function of $N^a_\tau-v^a\tau$, for different values of $N^a_+$, with $N^b_+=0$, $N^b_\tau-v^b\tau=0$ and $h=30s$. On the right, $\xi^a_\tau$ as a function of $N^b_\tau-v^b\tau$, for different values of $N^a_+$, with $N^b_+=0$, $N^a_\tau-v^a\tau=0$ and $h=30s$, $q=0.1$.}
\label{xiaer}
\end{figure}

\subsubsection{Impact of the risk aversion parameter and of the auction duration}\label{duration:graph}

We investigate the effect of the parameter $q$ which is the factor for the penalties received by the players for not reaching their trading targets and of the auction duration. 
We first consider $q=0.1$ and plot on the left side of Figure \ref{hnanb} $\xi^a_\tau$ as a function of $h$, for $N^a_\tau-v^a\tau=N^b_\tau-v^b\tau=0$, $N^b_+=0$ and for multiple values of $N^a_+$. On the right side of Figure \ref{hnanb}, we display $\xi^a_\tau$ as a function of $h$, with $N^a_\tau-v^a\tau=N^b_\tau-v^b\tau=0$, $N^a_+=0$ and for multiple values of $N^b_+$.
In Figure \ref{hnanb2} we fix $q=0.01$ and in Figure \ref{hnanb3} $q=0.001$.\\

\begin{figure}
\centering
\begin{subfigure}{\textwidth}
\centering
\includegraphics[width=15cm]{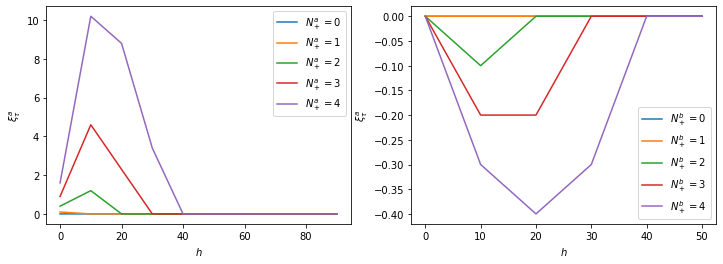}
\caption{$q=0.1$.}
\label{hnanb}
\end{subfigure}
\newline
\begin{subfigure}{\textwidth}
\centering
\includegraphics[width=15cm]{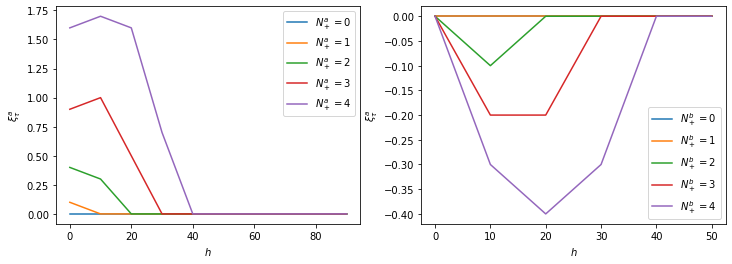}
\caption{$q=0.01$.}
\label{hnanb2}
\end{subfigure}
\newline
\begin{subfigure}{\textwidth}
\centering
\includegraphics[width=15cm]{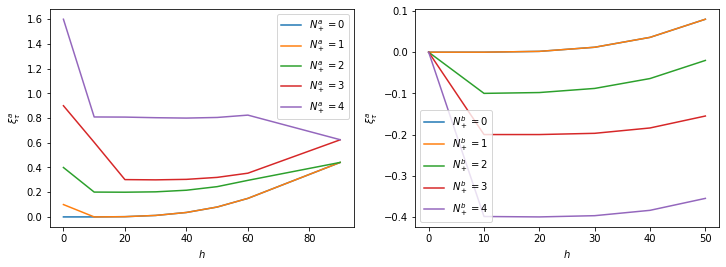}
\caption{$q=0.001$}
\label{hnanb3}
\end{subfigure}
\caption{$\xi^a_\tau$ as a function of $h$, with $N^a_\tau-v^a\tau=0$, $N^b_\tau-v^b\tau=0$. On the left, we fix $N^b_+=0$ and consider multiple values of $N^a_+$. On the right, we fix $N^a_+=0$ and consider multiple values of $N^b_+$.}
\label{fig:three_hnanb}
\end{figure}

We see in Figures \ref{hnanb} and \ref{hnanb2} that $\xi^a_\tau=0$ for $h$ large enough, which is no longer the case in Figure \ref{hnanb3}. This is because when the commitment to the target is severe, over a quite long time period both traders send on average the same number of orders as $v^a=v^b$ and the effect of $N^a_+$ or $N^b_+$ vanishes. We also observe that too short auctions may create some kind of arbitrage opportunities: the trader who triggers an auction is committed to trade at least a given volume. The other trader might choose to trade less to take advantage of the price imbalance in the auction, as the penalty he will have to pay will not be too large. Let us take the example of $h=20s$. In that case, the target is two lots for both Player $a$ and Player $b$. If Player $a$ triggers the auction with $N^a_+=4$ then Player $b$ will put a volume of $2$ in the auction meeting his target or perhaps even less (volume of $1$) meeting partially his target but benefiting from price impact. Such phenomenon is magnified in a situation as in Figure \ref{hnanb3} where the target commitment is very weak. In that case, both players try to benefit from price impact leading to a game where they both put smaller volumes than their target. For example, we see that the effect of the initial volume $N^a_+=1$ \textit{vs} $N^a_+=2$ takes more than $80$ seconds to vanish in Figure \ref{hnanb3}, although the target is $8$ lots for $80s$. This means that between $0$ and $80$ seconds, both investors play strategically to benefit from the effect of volume imbalance on the clearing price. \\

This shows that the duration of the auction should be large enough and related to reasonable practical values for $q$. Considering the auction duration helps to convey information to market participants, it should also probably depend on the deviation between the previous clearing price and the best offer price in the order book at the beginning of the auction. The larger the deviation, the longer the duration of the auction. Accurate duration calibration is left for further research.\\

We use the results of this section to choose suitable parameters for our study of the entire \textit{ad-hoc} auction in the next section. 
 
\subsection{Assessment of \textit{ad-hoc} auctions}
We now investigate the whole mechanism of \textit{ad-hoc} auctions and compare it with the classical CLOB and periodic auctions. We use Algorithm \ref{algovaluediscretegeneral} with a small timestep $\delta=0.05s$ and write $V^i$ for $J^i(.,(\tau^{i,*},\lambda^{i,*}),(\tau^{i,*},\lambda^{i,*}))$ for $i=a,b$.

\subsubsection{Choice of the parameters for the simulation study} The values for $v^a$ and $v^b$ will be of order $0.1$, so we expect roughly $2$ trades every $10$ seconds, which corresponds to the case of reasonably liquid assets. We fix $T=100s$ so that $T$ is large compared to the average time between trades. We take $q=0.01, h=20s$ and $\hat n=n_{ab}=3$. The justification for the relevance of these parameters is the following: 
\begin{itemize}
\item We have $\hat n>(v^a\wedge v^b) h$. This ensures that transactions occur both in the continuous and auction phases. As a matter of fact, if $\hat n<(v^a\wedge v^b) h$, the triggering cost for an auction is quite negligible with respect to the target amount within the auction. We numerically observe that in that case, investors do not use the continuous phase and trade only in the auction phases, which means that \textit{ad-hoc} auctions are reduced to periodic auctions.
\item Consider an auction triggered because both players are slightly behind their targets so that one of them, say Player $a$, triggers the auction and both should trade $3$ lots during the $20$ seconds. Then, suppose that Player $b$ tries to benefit from the price impact and trade only $2$ during the auction instead of $3$. Under these parameters, the price impact benefit of Player $b$ (which is equal to $2\times 1/K$) is exactly the cost paid for not meeting the target (which is equal to $qh$). Hence from the investors' viewpoint, these parameters correspond to reasonable balance between trading costs and target deviation penalties.
\end{itemize}

\subsubsection{Effect of $\hat n$ compared to $v^a$ and $v^b$}
Here we replace $\hat n_{ab}$ by a random variable which is so that if there is simultaneous triggering, it is attributed to Player $a$ or Player $b$ with probability $1/2$ and a volume commitment equal to $3$. We plot in Figure \ref{v_et_t0101p50}, Figure \ref{v_et_t01005} and Figure \ref{v_et_t01015} the values of the game at the origin and the average duration\footnote{In fact we only compute a proxy of the average duration. The computation details are given in Algorithm \ref{algo:meantime} in Appendix \ref{section:algo_duration}.} of the continuous phase at time $t=0$ for different values of $v^a$ and $v^b$.

\paragraph{Figure \ref{v_et_t0101p50} : supply and demand of similar order.} When $v^a=v^b=0.1$, the average duration of the continuous phase is $21$ seconds if the initial trading price $P$ is equal to $P^*_0$. We observe that we obviously get a symmetric average duration of the continuous trading phase with respect to the sign of $P-P^*$. The average duration of the continuous phase is maximal at $P=P
^*$, then decreases and becomes stationary.
This can be explained as follows: if $P$ is close to $P^*$ and because of the symmetry of $P^*$, locally around $t=0$, we expect to have oscillations of $P^*$ around $P$ so that either Player $a$ or Player $b$ can trade with the market maker. If $P^*$ increases (resp. if $P^*$ decreases) significantly beyond $P$, Player $b$ (resp. Player $a$) is likely to start an auction since his probability to trade with the market maker in a short amount of time becomes smaller. Although he can trade, the other player may also wish to trigger an auction, as his trading price becomes very unfavourable. So we see that \textit{ad-hoc} auctions ensure that the trading price does not deviate too much from the efficient price.\\

The plots of $V^a$ and $V^b$ in Figure \ref{v_et_t0101p50} are not even functions with respect to $P-P^*$. To explain it, take for example the position of Player $a$. If $P>P^*$ he can buy and launch an auction while if $P<P^*$ he can only launch an auction. Consequently, the situation $P>P^*$ is somehow more acceptable for him. We recall that as Player $a$ minimises and Player $b$ maximises, the value functions are symmetric with respect to the origin. We also see that for Player $a$, there is a large peak of the value function when $P-P^*$ is slightly negative and only a small downward bump when it is slightly positive. This means that for Player $a$, there is much more to lose when Player $b$ can trade with the market maker than to earn when he can trade with the market maker. This will be also confirmed in Table \ref{tab::compa_} below. 

\begin{figure}[!h]
\centering
\includegraphics[width=15cm]{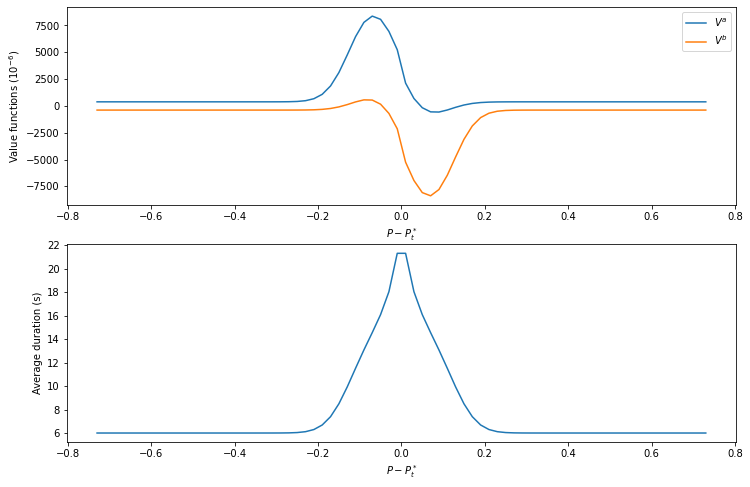}
\caption{Values of the game and average duration of the continuous phase, at time $t=0$, as functions of $P-P^*_0$, with $v^a=v^b=0.1$.}
\label{v_et_t0101p50}
\end{figure}

\paragraph{Figure \ref{v_et_t01005} : demand higher than supply for small investors.} When $v^a=0.1$ and $v^b=0.05$, Player $b$ is better off than Player $a$. When $P>P^*$, Player $a$ can trade with the market maker hence reducing the imbalance with respect to the volume of Player $b$. Player $b$ will typically not immediately trigger an auction because of the quite significant entry cost of the auction $\hat n=3$. This explains the quite long duration of the auction phase in this situation and the downward peak of the value function of Player $a$. When $P-P^*$ is negative, Player $b$ can trade with the market maker which could lead to an even larger imbalance from Player $a$'s perspective. Thus we expect Player $a$ to trigger the auction in that case explaining the short length of the continuous phase and the flat behaviour of the value functions on the left of $0$ (whatever $P-P^*<0$, Player $a$ will trigger an auction). 

 \begin{figure}[!h]
\centering
\includegraphics[width=15cm]{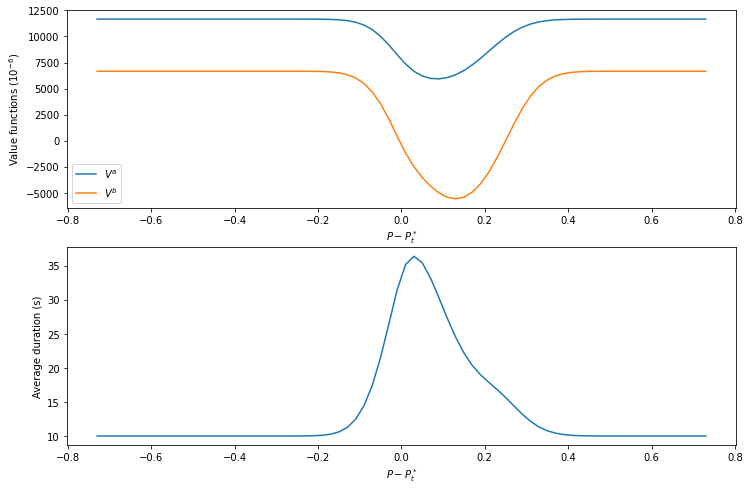}
\caption{Values of the game and average duration of the continuous phase, at time $t=0$, as functions of $P-P^*_0$, with $v^a=0.1$ and $v^b=0.05$.}
\label{v_et_t01005}
\end{figure}

\paragraph{Figure \ref{v_et_t01015} : supply higher than demand for large investors.} When $v^a=0.1$ and $v^b=0.15$ the situation differs significantly. When $P>P^*$, as previously, Player $a$ can trade with the market maker, which improves even more its imbalance position with respect to the volume of Player $b$ and it is particularly interesting when $P$ is only slightly larger than $P^*$. In that case, Player $b$ rapidly triggers an auction to prevent Player $a$ from trading. Note that contrary to the previous situation, the entry cost is not prohibitive here for Player $b$ as $v^b=0.15$. When $P$ is significantly larger than $P^*$, the price becomes too bad for Player $a$ who stops trading. Then a gaming situation occurs between the two players explaining the delay before one of them triggers the auction. Regarding the value functions, the peak of the orange graph is explained by the fact that it is very interesting for Player $b$ to trade with the market maker to reduce his imbalance with respect to the volume of Player $a$ (who may be reluctant to trigger an auction as $v^a$ is not very large).  

\begin{figure}[!h]
\centering
\includegraphics[width=15cm]{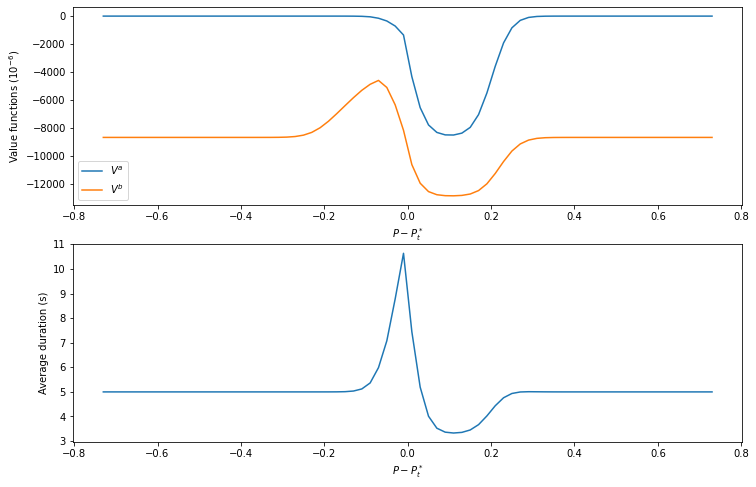}
\caption{Values of the game and average duration of the continuous phase, at time $t=0$, as functions of $P-P^*_0$, with $v^a=0.1$ and $v^b=0.15$.}
\label{v_et_t01015}
\end{figure}

\subsubsection{Comparison with periodic auctions and CLOB}

We finally provide the value functions and average durations in the case of {\it ad-hoc} auctions,  expensive periodic auctions ($\hat n= 3$ and no trading allowed in the continuous phase), inexpensive periodic auctions ($\hat n=1$ and no trading allowed in the continuous phase) and CLOB. In the case of CLOB, the players trade only with the market maker and pay $1/K$ for each trade. The average duration is then defined as the average time between two trades and the value as the amount paid per unit of time. The results are shown in Table \ref{tab::compa_}.

\begin{table}[h!]
\centering
    \begin{tabular}{|c || c c | c | c | c c | c | c||} 
         \hline
        &\multicolumn{4}{c|}{$V^a$ (1e-6)}&\multicolumn{4}{c||}{Average duration}\\
        \hline
        Market design & \multicolumn{2}{c|}{\shortstack{$h = 20$,\\ $\hat{n}=3$}}&\shortstack{$h = 20$,\\ $\hat{n}=1$}&CLOB& \multicolumn{2}{c|}{\shortstack{$h = 20$,\\ $\hat{n}=3$}}&\shortstack{$h = 20$,\\ $\hat{n}=1$}&CLOB\\
        \hline
         continuous trading allowed&Yes&No&No&No&Yes&No&No&No  \\ [0.9ex] 
         \hline\hline
         $v^a=0.1$, $v^b=0.1$&3685.6&384.6&0.0&10000.0&21.3s&6.0s&0.0s&10.0s \\ 
         \hline
         $v^a=0.05$, $v^b=0.1$&392.8&-6666.7&-5000.0&5000.0&33.3s&10.0s&0.0s&20.0s \\  
         \hline
         $v^a=0.1$, $v^b=0.05$&7800.0&11606.7&10000.0&10000.0&33.3s&10.0s&0.0s&10.0s \\  
         \hline
         $v^a=0.15$, $v^b=0.1$&9397.9&8680.0&10000.0&15000.0&9.0s&5.0s&0.0s&6.7s \\ 
         \hline
         $v^a=0.1$, $v^b=0.15$&-2841.8&0.0&0.0&10000.0&9.0s&5.0s&0.0s&10s \\ 
         \hline
    \end{tabular}
    \caption{$V^a$ and average duration of the continuous trading phase for different values of $v^a$ and $v^b$ with $q=0.01$.}
    \label{tab::compa_}
\end{table}

We notice first that, if continuous trading with the market maker is allowed, the average duration of the pre-auction phase is longer. This is because both players try to trade with the market maker if possible in order to push the settlement price of the next auction in their favour.\\ 

If $v^a=v^b=0.1$, Player $a$ prefers the case where there is no continuous trading. This is in agreement with our interpretation of Figure \ref{v_et_t0101p50} since Player $a$ has much more to lose when Player $b$ can trade with the market maker than to earn when he can trade with the market maker. Moreover, if the triggering volume is small ($\hat n=1$), the probability of the auctions to be balanced is large and the player who cannot trade with the market maker triggers an auction quickly. The case $\hat{n}=3$ provides an intermediary between periodic auctions and CLOB in terms of value functions.\\

If $v^a$ and $v^b$ are small and asymmetric (either $v^a=0.05$ and $v^b = 0.1$ or $v^a=0.1$ and $v^b=0.05$), we observe that the player with the larger target benefits from \textit{ad-hoc} auctions. We explain this as follows: if the larger player can trade with the market maker, he is able to liquidate his temporary surplus at a low cost with the market maker and so suffers less from price impact in the auction, which is more balanced than in the situation without continuous trading. In this case, it is too costly for the smaller player to trigger an auction since $\hat n$ is too high compared to the target $0.05$. The larger player is thus the first to trigger the auction if the price becomes too unfavourable, in a way signalling to the smaller player that it is preferable to trade at the forthcoming auction instead of at the clearing price. Otherwise, if the smaller player trades with the market maker during the continuous trading phase, the larger player triggers the auction to protect himself from an excessively unfavourable price at the auction. The smaller player benefits from information leakage/market impact generated by the larger player, while the larger player uses his informational advantage of being the larger player by capturing mistimed liquidity from the smaller player. In both cases, the larger player is the one triggering the auction and benefits from the continuous trading phase. This is in agreement with Figure \ref{v_et_t01005} where $V^a$ takes its lowest value for $P>P^*$ with $P$ close to $P^*$. In addition, compared to the case without market maker or with $|P-P^*|$ large, the temporary target imbalance has less impact on the distance between the clearing price and the efficient price. This is a direct consequence of the surplus of orders from the larger player being absorbed by the market maker. We consider this an advantage of \textit{ad-hoc} auctions: the clearing price has less volatility.\\

We now turn to $v^a=0.1$ and $v^b=0.15$. As before, if Player $a$ (smaller player) trades with the market maker, he can liquidate part of his volume but Player $b$ (larger player) quickly triggers an auction to prevent him from doing so. The larger player can indeed trigger the auction since $\hat n=3$ coincides with his target. When Player $b$ trades with the market maker, unlike the previous case, the auction triggering cost is reasonable for Player $a$. The continuous phase appears as an opportunity for Player $a$ to prevent Player $b$ from mitigating his inventory since in this case Player $a$ triggers the auction. This is in accordance with Figure \ref{v_et_t01015} above. Conversely, we observe that the value functions of Player $b$ are quite similar considering \textit{ad-hoc} auctions or classical periodic auctions. One conclusion is that for large investors, the smaller one benefits a lot from \textit{ad-hoc} auctions compared to periodic auctions and CLOB, while the larger one is quite indifferent between \textit{ad-hoc} and periodic auctions.\\

The parameter $q$ plays quite an important role since it dictates the probability of an auction to be balanced out. We refer to Appendix \ref{app:q005} for the value functions and average durations with $q=0.005$.  For this value of the penalties, auctions are rarely balanced and the CLOB design becomes more relevant. In this configuration, allowing continuous trading is always beneficial if $\hat{n}=3$ since it mitigates price impact during the auction. The value functions are close to those observed with periodic auctions but have the attractive property of having very long periods of continuous trading: the price remains constant for a long time while with periodic auctions, auctions are triggered as soon as someone needs to trade. Also the larger player still benefits a lot from being able to trade with the market maker.

\section*{Acknowledgments}
The authors gratefully acknowledge the financial support of the ERC Grant 679836 Staqamof and the Chaire {\it Analytics and Models for Regulation}. They are also thankful to Alexandra Givry, Iris Lucas and Eric Va for insightful discussions.

\newpage
\bibliographystyle{apalike}
\bibliography{b}

\newpage

\newpage
\begin{appendix}
\label{appdx}
\section{Proof of Proposition \ref{prop:DPP}}\label{app:proofpropsub} 
The proof of the existence of an open-loop Nash equilibrium for the sub-game \eqref{sub-game} is a direct extension of the results of \cite{hamadne2014bangbang} and \cite{jusselin2019optimal}, taking into consideration the continuous trading phase, together with a smooth decomposition of the value function at the optimum.\\

We focus on the dynamic programming principle \eqref{DPP}. We follow the same argument as in \cite{cvitanic1993} Proposition 6.2 or \cite{euch2018optimal} Lemma A.4.
First, let us write
$\chi^a: = \mathcal{C}^a_{auc}+\frac{N^a_{\tau,\tau+h}(N^a_{\tau,\tau+h}-N^b_{\tau,\tau+h})}{K}.$
From the definition of an OLNE, we have
\begin{equation*}
    J^a(x_0,(\tau^{a,*},\lambda^{a,*}),(\tau^{b,*},\lambda^{b,*})) = \underset{\lambda^a\in\mathcal{U}_{[0,\tau]}}{\inf}\; \mathbb{E}^{\lambda^a,\lambda^{b,*}}\Big[\frac{L^a_\tau+\chi^a}{\tau+h}\Big].
\end{equation*}
Using the tower property we get
\begin{equation}\label{TowerProp}
    J^a(x_0,(\tau^{a,*},\lambda^{a,*}),(\tau^{b,*},\lambda^{b,*})) =\underset{\lambda^a\in\mathcal{U}_{[0,\tau]}}{\inf}\;\mathbb{E}^{\lambda^a,\lambda^{b,*}}\Big[\frac{L^a_\tau+\mathbb{E}^{\lambda^a,\lambda^{b,*}}_\tau[\chi^a]}{\tau+h}\Big].
\end{equation}
Moreover applying Bayes' formula leads to
\begin{equation}\label{Bayes}
\begin{split}
    \mathbb{E}^{\lambda^a,\lambda^{b,*}}_\tau[\chi^a] &= \mathbb{E}^{\lambda^0,\lambda^0}_\tau\big[\frac{\Psi_{T+h}^{\lambda^a,\lambda^{b,*}}}{\Psi_\tau^{\lambda^a,\lambda^{b,*}}}\chi^a\big]=\mathbb{E}^{\lambda^{a},\lambda^{b,*}_{[\tau,T+h]}}_\tau[\chi^a]\geq \underset{\mu^a\in\mathcal{U}_{[\tau,T+h]}}{\text{essinf}}\; \mathbb{E}^{\mu^{a},\lambda^{b,*}_{[\tau,T+h]}}_\tau[\chi^a].
\end{split}
\end{equation}
Therefore, using both \eqref{TowerProp} and \eqref{Bayes} we obtain

\begin{equation}\label{DPPgeq}
  J^a(x_0,(\tau^{a,*},\lambda^{a,*}),(\tau^{b,*},\lambda^{b,*}))\geq \underset{\lambda^a\in\mathcal{U}_{[0,\tau]}}{\inf} \; \mathbb{E}^{\lambda^a,\lambda^{b,*}}\Big[\frac{L^a_\tau+\xi_\tau^a}{\tau+h}\Big].
\end{equation}
Conversely, let $\lambda^a\in\mathcal{U}_{[0,\tau]}$ and $\mu^a\in\mathcal{U}_{[\tau,T+h]}$. Recalling the definition $(\lambda^a\otimes_\tau\mu^a)_u := \lambda^a_u\mathbf{1}_{u\leq\tau}+\mu^a_u\mathbf{1}_{\tau<u}$, we get that $\lambda^a\otimes_\tau\mu^a\in\mathcal{U}$ and
\begin{equation}\label{inegNashleq}
\begin{split}
    J^a(x_0,(\tau^{a,*},\lambda^{a,*}),(\tau^{b,*},\lambda^{b,*}))&\leq J^a(x_0,(\tau^{a,*},\lambda^a\otimes_\tau\mu^a),(\tau^{b,*},\lambda^{b,*}))\\
    &=\mathbb{E}^{\lambda^a\otimes_\tau\mu^a,\lambda^{b,*}}\Big[\frac{L^a_\tau+\chi^a}{\tau+h}\Big]\\
    &=\mathbb{E}^{\lambda^a\otimes_\tau\mu^a,\lambda^{b,*}}\Big[\frac{L^a_\tau+\mathbb{E}^{\lambda^a\otimes_\tau\mu^a,\lambda^{b,*}}_\tau[\chi^a]}{\tau+h}\Big].
\end{split}
\end{equation}
Thus, remarking that $\frac{\Psi_{T+h}^{\lambda^a\otimes_\tau\mu^a,\lambda^{b,*}}}{\Psi_\tau^{\lambda^a\otimes_\tau\mu^a,\lambda^{b,*}}}=\Psi_{\tau,T+h}^{\mu^a,\lambda^{b,*}}$, we deduce
\begin{equation*}
    \mathbb{E}^{\lambda^a\otimes_\tau\mu^a,\lambda^{b,*}}_\tau[\chi^a]=\mathbb{E}^{0,0}_\tau\Big[\Psi_{\tau,T+h}^{\mu^a,\lambda^{b,*}} \chi^a\Big]
    =\mathbb{E}_\tau^{\mu^a,\lambda^{b,*}_{[\tau,T+h]}}[\chi^a].
\end{equation*}
Using Lemma A.3 of \cite{euch2018optimal} (extended to stopping times), we can build a sequence $(\mu^a_n)_{n\in\mathbb{N}}$ such that 
\begin{equation}\label{converge:xitau}
    \mathbb{E}_\tau^{\mu^a_n,\lambda^{b,*}_{[\tau,T+h]}}[\chi^a]\searrow \xi^a_\tau.
\end{equation}
By the monotonous convergence theorem together with \eqref{inegNashleq} and \eqref{converge:xitau}, we obtain
\begin{equation}\label{DPPleq}
    J^a(x_0,(\tau^{a,*},\lambda^{a,*}),(\tau^{b,*},\lambda^{b,*})) \leq \underset{\lambda^a\in\mathcal{U}_{[0,\tau]}}{\hbox{inf}}\;\mathbb{E}^{\lambda^a,\lambda^{b,*}}\Big[\frac{L^a_\tau+\xi^a_\tau}{\tau+h}\Big].
\end{equation}
We can prove similar results for Player $b$. We conclude from  \eqref{DPPgeq} and \eqref{DPPleq} and the corresponding inequalities for Player $b$ that the dynamic programming principle \eqref{DPP} holds.\\

Finally, let $((\tau^{a,*},\lambda^{a,*}),(\tau^{b,*},\lambda^{b,*}))$ be an OLNE. Then, from Definition \ref{def:nash1}, for any $\lambda^a\in \mathcal U$ we have  
 \begin{equation}\label{inegNashcorollary}
 J^a(x_0,(\tau^{a,*},\lambda^{a,*}),(\tau^{b,*},\lambda^{b,*}))\leq J^a(x_0,(\tau^{a,*},\lambda^a),(\tau^{b,*},\lambda^{b,*})).
 \end{equation}
 Assume that there exists $\mu^a\neq \lambda^{a,*}_{\tau,T+h}$ such that
 \[ \mathbb{E}^{\mu^{a},\lambda^{b,*}_{[\tau,T+h]}}_\tau\Big[\mathcal{C}^a_{auc}+\frac{N^a_{\tau,\tau+h}(N^a_{\tau,\tau+h}-N^b_{\tau,\tau+h})}{K}\Big] < \mathbb{E}^{\lambda^{a,*}_{[\tau,T+h]},\lambda^{b,*}_{[\tau,T+h]}}_\tau\Big[\mathcal{C}^a_{auc}+\frac{N^a_{\tau,\tau+h}(N^a_{\tau,\tau+h}-N^b_{\tau,\tau+h})}{K}\Big] .\]
 Then
 \[\xi_{\tau}^a< \mathbb{E}^{\lambda^{a,*}_{[\tau,T+h]},\lambda^{b,*}_{[\tau,T+h]}}_\tau\Big[\mathcal{C}^a_{auc}+\frac{N^a_{\tau,\tau+h}(N^a_{\tau,\tau+h}-N^b_{\tau,\tau+h})}{K}\Big] .\]
 Let $\lambda^a\in \mathcal U_{[0,\tau]}$, we have
 \[ \mathbb{E}^{\lambda^a,\lambda^{b,*}}\big[\frac{L^a_\tau+\xi^a_\tau}{\tau+h}\big] < \mathbb{E}^{\lambda^a,\lambda^{b,*}}\Big[\frac{L^a_\tau+\mathbb{E}^{\lambda^{a,*}_{[\tau,T+h]},\lambda^{b,*}_{[\tau,T+h]}}_\tau\Big[\mathcal{C}^a_{auc}+\frac{N^a_{\tau,\tau+h}(N^a_{\tau,\tau+h}-N^b_{\tau,\tau+h})}{K}\Big]}{\tau+h}\Big]. \]
Therefore
 \[\inf_{\lambda^a\in \mathcal U_{[0,\tau]}} \mathbb{E}^{\lambda^a,\lambda^{b,*}}\big[\frac{L^a_\tau+\xi^a_\tau}{\tau+h}\big] < \inf_{\substack{\lambda^a \in \mathcal U_{[0,\tau]},\\ \lambda_t^a\mathbf{1}_{\tau<t\leq T+h}= (\lambda^{a,*}_{[\tau,T]})_t}} \mathbb{E}^{\lambda^a,\lambda^{b,*}}\Big[\frac{L^a_\tau+\mathbb{E}^{\lambda^{a,*}_{[\tau,T+h]},\lambda^{b,*}_{[\tau,T+h]}}_\tau\Big[\mathcal{C}^a_{auc}+\frac{N^a_{\tau,\tau+h}(N^a_{\tau,\tau+h}-N^b_{\tau,\tau+h})}{K}\Big]}{\tau+h}\Big].\]
 Using the dynamic programming principle \eqref{DPP} together with \eqref{inegNashcorollary} we deduce
 \[ J^a(x_0,(\tau^{a,*},\lambda^{a,*}),(\tau^{b,*},\lambda^{b,*})) < J^a(x_0,(\tau^{a,*},\lambda^{a,*}),(\tau^{b,*},\lambda^{b,*})) ,\]
 leading to a contradiction. We have a similar result for Player $b$. We conclude that $(\lambda^{a,*}_{[\tau,T+h]},\lambda^{b,*}_{[\tau,T+h]})$ is an open-loop Nash equilibrium for the sub-game \eqref{sub-game}.
 
 \section{Existence of an OLNE for $[0,T]-$valued stopping times: a verification theorem}\label{app:verif}

We need to extend the results of \cite{basei2016nonzerosum} and \cite{basei2019nonzerosum} to include jump processes and expectations given by non-trivial risk measures.\\

Following the ideas of \cite{jusselin2019optimal}, we turn to the definition of the Hamiltonian related to the optimisation of the players during the auction. For any $z\in \mathbb{R}^2$ and $\epsilon_a,\epsilon_b\in[\lambda_-,\lambda_+]$, we set
\begin{equation*}
    \begin{split}
        \lambda^*_a(z,\epsilon_a) =& 1_{z_1>0}\lambda_-+1_{z_1<0}\lambda_++1_{z_1=0}\epsilon_a\\
        \lambda^*_b(z,\epsilon_b) =& 1_{z_2<0}\lambda_-+1_{z_2>0}\lambda_++1_{z_2=0}\epsilon_b.
    \end{split}
\end{equation*}
As $z_1\lambda^*_a(z,\epsilon_a)$ and $z_2\lambda^*_b(z,\epsilon_b)$ do not depend on $\epsilon_a$ and $\epsilon_b$, we simply denote them by $z_1\lambda^*_a(z)$ and $z_2\lambda^*_b(z)$.
For any $z,\Tilde{z}\in \mathbb{R}^2$ and any $\epsilon\in[\lambda_-,\lambda_+]$, we set 
\begin{equation*}
    \begin{split}
        H^{a,*}(p^*,z,\Tilde{z},\epsilon) =&1_{p^*< P}z_1\lambda^*_a(z) +1_{p^*> P}z_2\lambda^*_b(\Tilde{z},\epsilon)\\
        H^{b,*}(p^*,z,\Tilde{z},\epsilon) =&1_{p^*>P}z_2\lambda^*_b(z) +1_{p^*<P}z_1\lambda^*_a(\Tilde{z},\epsilon).
    \end{split}
\end{equation*}

We define for any $j\in \{a,b\}$
\begin{itemize}
\item  $\mathcal L^j$ from $[0,T]\times \mathbb R\times \mathbb N\times \mathbb N\times \mathbb R\times \mathbb R^+\times \mathbb R^+\times \mathbb R\times \mathbb R\times [\lambda_-,\lambda_+]$ into $\mathbb R$ by
\[\mathcal L^j(t,p^*,n^a,n^b,d_t,d_2,\ell^a,\ell^b,d,\tilde d,\epsilon):=d_t+\frac12\sigma^2 d_2+ q\sum_{i\in \{a,b\}}(v^it-n^i)^2\ell^i+H^{j,*}(p^*,d,\tilde d,\epsilon), \]
where $(t,p^*,n^a,n^b,d_t,d_2,\ell^a,\ell^b,d,\tilde d,\epsilon)\in [0,T]\times \mathbb N\times \mathbb N  \times \mathbb R\times \mathbb R\times \mathbb R\times \mathbb R^+\times \mathbb R^+\times \mathbb R\times \mathbb R\times [\lambda_-,\lambda_+].$
\item the function $G^j$ from $[0,T]\times \mathbb R\times\mathbb R\times\mathbb N\times \mathbb N$ into $\mathbb R$ by 
\[ G^j(t,l^j,n^a,n^b):=\frac{\beta_j l^j+g_j^{\text{first}}(t,n^a,n^b)}{t+h},\] where $\beta_j=\mathbf 1_{j=a}-\mathbf 1_{j=b}.$

\item for any map $U: (t,p^*,l^a,l^b,n^a,n^b)\in [0,T]\times\mathbb{R}\times\mathbb{R}_+\times\mathbb{R}_+\times\mathbb{N}\times\mathbb{N}\longrightarrow \mathbb{R}$, the domains
\[\Gamma^j(U)=\{ (t,p^*,n^a,n^b,l^a,l^b)\in [0,T)\times \mathbb R\times\mathbb N\times \mathbb N\times \mathbb R\times \mathbb R,\; \beta_j G^j> \beta_j U\},\]
\[\partial \Gamma^j(U)=\{ (t,p^*,n^a,n^b,l^a,l^b)\in [0,T)\times \mathbb R\times\mathbb N\times \mathbb N\times \mathbb R\times \mathbb R,\; G^j=U\}.\]
together with its derivative operators 
\begin{align*}
    &D^aU(t,p^*,l^a,l^b,n^a,n^b) = U(t,p^*,l^a+(P-p^*),l^b,n^a+1,n^b)-U(t,p^*,l^a,l^b,n^a,n^b)\\
    &D^bU(t,p^*,l^a,l^b,n^a,n^b) = U(t,p^*,l^a,l^b+(p^*-P),n^a,n^b+1)-U(t,p^*,l^a,l^b,n^a,n^b)\\
    &DU = (D^aU,D^bU)^T.
\end{align*}

\end{itemize}
The quantity $D^j U$ describes the change in the value of the process $(U(t,P^*_t,L^a_t,L^b_t,N^a_t,N^b_t))_{t\in[0,T]}$ when Player $j$, $j\in \{a,b\}$, sends an order which triggers a trade at the fixed price $P$. The set $\Gamma^j(U)$ is the domain on which Player $j$ would rather have a game of value $U$ than trigger an auction alone (and thus have a game of value $G^j$). The set $\partial \Gamma^j(U)$ is the domain on which he is indifferent. We have the following result.

\begin{theorem}\label{thm:verifN1}
Let $V^a$ and $V^b$ be two functions of $(t,p^*,l^a,l^b,n^a,n^b)$ from $[0,T]\times\mathbb{R}\times\mathbb{R}_+\times\mathbb{R}_+\times\mathbb{N}\times\mathbb{N}$ into $\mathbb{R}$. Assume that there exist two maps $\varepsilon_a$, $\varepsilon_b$ from $\mathbb{R}_+\times\mathbb{R}\times\mathbb{R}_+\times\mathbb{R}_+\times\mathbb{N}\times\mathbb{N}$ into $[\lambda_-,\lambda_+]$ such that
\begin{enumerate}[(i)]
    \item $V^a$ and $V^b$ are $\mathcal{C}^1$ in time on $[0,T)$ and in their third (on $\mathbb{R}$) and fourth arguments (on $\mathbb{R}$), $\mathcal{C}^2$ in their second argument (on $\mathbb{R}$), and are solutions to the following variational system:
\begin{equation}\label{sys_nash}
    \begin{cases}
        \max\big\{ -\mathcal L^a(t,p^*,n^a,n^b,\partial_t V^a,\partial^2_{pp}V^a,\partial_{l^a}V^a,\partial_{l^b}V^a,DV^a,DV^b,\varepsilon_b),V^a-G^a\big\} =0\text{, on } \Gamma^b(V^b) \\[0.8em]
        \frac{l^a+g^{\text{second}}_a(t,n^a,n^b)}{t+h} = V^a\text{ on }\partial \Gamma^b(V^b)\\[0.8em]
        V^a(T,p^*,l^a,l^b,n^a,n^b)=\frac{l^a+g^{\text{T}}_a(T,n^a,n^b)}{T+h} \\[0.8em]
        \min\big\{-\mathcal L^b(t,p^*,n^a,n^b,\partial_t V^b,\partial^2_{pp}V^b,\partial_{l^a}V^b,\partial_{l^b}V^b,DV^b,DV^a,\varepsilon_a),V^b-G^b\big\} =0\text{, on } \Gamma^a(V^a) \\[0.8em]
        \frac{-l^b+g^{\text{second}}_b(t,n^a,n^b)}{t+h} = V^b\text{ on }\partial \Gamma^a(V^a) \\[0.8em]
        V^b(T,p^*,l^a,l^b,n^a,n^b) =\frac{-l^b+g^{\text{T}}_b(T,n^a,n^b)}{T+h}.
    \end{cases}
\end{equation}
    \item $g^{\text{second}}_a\leq g^{\text{sim}}_a$ on $\partial \Gamma^b(V^b)$ and $g^{\text{second}}_b\geq g^{\text{sim}}_b$ on $\partial \Gamma^a(V^a)$. 
\end{enumerate}
Then\footnote{Here $\hat{\lambda}^a$ and $\hat{\lambda}^b$ denote the strategies played by Player $a$ and Player $b$ during the auction given by Proposition \ref{prop:DPP}.} $((\tau^a,\lambda_a^*(DV^a,\varepsilon_a)\otimes_\tau\hat{\lambda}^a),(\tau^b,\lambda_b^*(DV^b,\varepsilon_b)\otimes_\tau\hat{\lambda}^b))$ is an OLNE in the sense of Definition \ref{def:nash1}, where
\begin{equation*}
\begin{split}
    \tau^a&=\inf\{t\geq 0\text{, } (t,P^*_t,L^a_t,L^b_t,N^a_t,N^b_t)\in \partial \Gamma^a(V^a) \}\\
    \tau^b&=\inf\{t\geq 0\text{, } (t,P^*_t,L^a_t,L^b_t,N^a_t,N^b_t)\in \partial \Gamma^b(V^b) \}.
\end{split}
\end{equation*}
\end{theorem}

\begin{remark}\label{rem:condverif}
The differentiability conditions are very strong. In a non bang-bang case, extending the results of \cite{basei2016nonzerosum} to the case of jump processes, it is possible to show that $\mathcal{C}^1$-differentiability in the third and fourth arguments is enough if $\partial \Gamma^a(V^a)$ and $\partial \Gamma^b(V^b)$ are Lipschitz surfaces.
Nevertheless, note that in Theorem \ref{thm:verifN1}, Condition $(ii)$ is easier to meet when $\hat{n}=\hat{n}_{ab}=0$. This is because the stopping domains are allowed to intersect since for $i=a,b$ we have $g_i^{\text{first}}=g_i^{\text{sim}}=g_i^{\text{second}}$.
\end{remark}

\begin{remark}
Suppose that there exists a solution to System \eqref{sys_nash} and that the players play the associated Nash equilibrium. Then, as soon as $g_a^{\text{first}}\neq g_a^{\text{second}}$ and $g_b^{\text{first}}\neq g_b^{\text{second}}$, necessarily we must have
\[\partial \Gamma^a(V^a)\cap \partial \Gamma^b(V^b)=\emptyset,\] \textit{i.e.} the two players never trigger an auction at the same time. In practice, we have $g^{\text{first}}_a\neq g^{\text{second}}_a$ and $g^{\text{first}}_b\neq g^{\text{second}}_b$ if $\hat n\neq 0$ and $h$ is not too large. Also, note that from our numerical investigations, it seems that there is no uniqueness for the solution of System \eqref{sys_nash}.
\end{remark}

\begin{proof}[Proof of the verification theorem \ref{thm:verifN1}]
Suppose that the functions $V^a$ and $V^b$ satisfy the above conditions with the maps $\varepsilon_a$ and $\varepsilon_b$ and that Player $b$ plays the strategy $(\tau^b,\lambda_b^*(DV^b,\varepsilon_b)\otimes_\tau\hat{\lambda}^b)$. Using Proposition \ref{DPP}, we see that it is optimal for Player $a$ to play $\hat{\lambda}^a$ after $\tau$, thus obtaining $\xi^a_\tau$ at $\tau$. For $t\leq \tau^b$, we write 
\[V^{a,*}_t = \underset{\tau^{a,*}\in\mathcal{T}_{t,T},\lambda^a\in \mathcal U_{[t,T+h]}}{\einf} \;\mathbb{E}^{\lambda^a,\lambda^{b,*}}_t\big[\frac{L^a_\tau+\xi^a_\tau}{\tau+h}\big].\]
On $\partial \Gamma^b(V^b)$, as $t\leq \tau^b$ and because of the definition of $\tau^b$, we necessarily have $t=\tau^b$. So Player $b$ triggers an auction, from $(ii)$
\[V^{a,*}_t = \min(\frac{l^a+g^{\text{second}}_a(t,n^a,n^b)}{t+h},\frac{l^a+g^{\text{sim}}_a(t,n^a,n^b)}{t+h}) = \frac{l^a+g^{\text{second}}_a(t,n^a,n^b)}{t+h}\] and Player $a$ does not trigger an auction. On $\Gamma^b(V^b)$, necessarily $t<\tau^b$, \textit{i.e.} Player $b$ does not trigger an auction. Player $a$ then solves a classical optimal stopping problem. Standard dynamic programming arguments yield the quasi-variational equality
\begin{equation*}
    \max\big\{ -\mathcal L^a(t,p^*,n^a,n^b,\partial_t V,\partial^2_{pp}V,\partial_{l^a}V,\partial_{l^b}V,DV,DV^b,\varepsilon_b),V-G^a\big\} =0
\end{equation*}
for the value of Player $a$'s game.
We now detail the computation of the generator $\mathcal L^a$. As the problem depends only on $(t,P^*_t,N^a_t, N^b_t, L^a_t, L^b_t)$, Itô's formula provides the following expression
\[\partial_t V + \frac12\sigma^2 \partial_{pp}^2 V+ q\sum_{i\in \{a,b\}}(v^it-N^i_t)^2\partial_{l^i} V+\underset{\lambda^a\in[\lambda_-,\lambda_+]}{\inf} \lambda^a \mathbf{1}_{P>P^*_t} D^aV +  \lambda_b^*(DV^b,\varepsilon_b)\mathbf{1}_{P<P^*_t}D^bV. \]
The infimum is reached for $\lambda^a = \lambda_a^*(DV,\tilde\varepsilon_a)$, for any $\tilde \varepsilon_a\in[\lambda_-,\lambda_+]$ and \[\lambda_a^*(DV,\tilde\varepsilon_a) \mathbf{1}_{P>P^*_t} D^aV +  \lambda_b^*(DV^b,\varepsilon_b)\mathbf{1}_{P<P^*_t}D^bV = H^{a,*}(P^*_t,DV,DV^b,\varepsilon_b),\] which gives $\mathcal L^a$.\\

We conclude by using classical verification arguments for mixed optimal stopping problems (see for example \cite{booktouzi} or \cite{basei2016nonzerosum}). They enable us to show that necessarily, a classical solution $V^a$ of this quasi-variational equality is equal to the value of Player $a$'s game and that the optimal controls are given by the minimiser $\lambda^a = \lambda_a^*(DV^a,\tilde\varepsilon_a)$ (for any $\tilde \varepsilon_a\in[\lambda_-,\lambda_+]$) and the stopping time
\begin{align*}
    \tau^a=&\inf\{t\geq 0\text{, } V^a(t,P^*_t,L^a_t,L^b_t,N^a_t,N^b_t)\geq G^a(t,P^*_t,L^a_t,L^b_t,N^a_t,N^b_t) \}\\
    =&\inf\{t\geq 0\text{, } (t,P^*_t,L^a_t,L^b_t,N^a_t,N^b_t)\in \partial \Gamma^a(V^a) \}.
\end{align*}
\end{proof}
\section{Proofs of Section \ref{sec::discrete}}
 
\subsection{Proof of Theorem \ref{thm:backward:unimportant}}\label{proof:discrete:unimportant}
We prove by backward induction on $k\in\{\frac{T}{\delta},...,0\}$ that 
\begin{equation*}
\begin{cases}
    U^a_{k} = \underset{\tau^a\in\mathcal{T}^d_{k\delta,T}, \lambda^a\in\mathcal{U}_{[k\delta,T+h]}}{\einf}&\mathbb{E}^{\lambda^a,\Tilde{\lambda}^{b,*}_{k}}_{k\delta}\big[\frac{L^a_\tau+\mathcal{C}^a_{auc}+\frac{N^a_{\tau,\tau+h}(N^a_{\tau,\tau+h}-N^b_{\tau,\tau+h})}{K}}{\tau+h}\big]\\
    U^b_{k} = \underset{\tau^b\in\mathcal{T}^d_{k\delta,T}, \lambda^b\in\mathcal{U}_{[k\delta,T+h]}}{\esup} &\mathbb{E}^{\Tilde{\lambda}^{a,*}_{k},\lambda^b}_{k\delta}\big[\frac{-L^b_\tau-\mathcal{C}^b_{auc}+\frac{N^b_{\tau,\tau+h}(N^a_{\tau,\tau+h}-N^b_{\tau,\tau+h})}{K}}{\tau+h}\big]
\end{cases}
\end{equation*}
and that the above extrema are reached for $\tau^a=\tau^b = \tau^*_k$, $\lambda^a=\Tilde{\lambda}^{a,*}_k$ and $\lambda^b=\Tilde{\lambda}^{b,*}_k$ with 
\begin{equation*}
    \tau^{*}_k=\delta\inf\Big\{l\in\llbracket k,T/\delta\rrbracket\text{, }U^a_l =\frac{L^a_{l\delta}+g^{\text{first}}_a({l\delta},N^a_{l\delta},N^b_{l\delta})}{l\delta+h}\text{ or } U^b_l=\frac{-L^b_{l\delta}+g^{\text{first}}_b({l\delta},N^a_{l\delta},N^b_{l\delta})}{l\delta+h}\Big\},
\end{equation*}
$\Tilde{\lambda}^{a,*}_k=\big(\bigotimes_{l\in\llbracket k,T/\delta-1\rrbracket}\lambda^{a,*}_l\big)\otimes_{\tau^*}\hat{\lambda}^{a}$ and $\Tilde{\lambda}^{b,*}_k=\big(\bigotimes_{l\in\llbracket k,T/\delta-1\rrbracket}\lambda^{b,*}_l\big)\otimes_{\tau^*}\hat{\lambda}^{b}$. Applying this result at $k=0$ we will get that $(\tau^*,\Tilde{\lambda}^{a,*}),(\tau^*,\Tilde{\lambda}^{b,*})$ is an OLNED.
For $k=T/\delta$, the result follows directly from Proposition \ref{prop:DPP}. Suppose the result holds for $k+1\in\llbracket1,T/\delta\rrbracket$. We show that it holds for $k$. We thus assume that
\begin{equation*}
\begin{cases}
    U^a_{k+1} = \underset{\tau^a\in\mathcal{T}^d_{(k+1)\delta,T}, \lambda^a\in\mathcal{U}_{[(k+1)\delta,T+h]}}{\einf}&\mathbb{E}^{\lambda^a,\Tilde{\lambda}^{b,*}_{k+1}}_{(k+1)\delta}\big[\frac{L^a_\tau+\mathcal{C}^a_{auc}+\frac{N^a_{\tau,\tau+h}(N^a_{\tau,\tau+h}-N^b_{\tau,\tau+h})}{K}}{\tau+h}\big]\\
    U^b_{k+1} = \underset{\tau^b\in\mathcal{T}^d_{(k+1)\delta,T}, \lambda^b\in\mathcal{U}_{[(k+1)\delta,T+h]}}{\esup} &\mathbb{E}^{\Tilde{\lambda}^{a,*}_{k+1},\lambda^b}_{(k+1)\delta}\big[\frac{-L^b_\tau-\mathcal{C}^b_{auc}+\frac{N^b_{\tau,\tau+h}(N^a_{\tau,\tau+h}-N^b_{\tau,\tau+h})}{K}}{\tau+h}\big].
\end{cases}
\end{equation*}
From the definition of $\lambda^{a,*}_k$ together with our induction assumption above, we have
\begin{equation}\label{eq::reecriture_U}
\begin{split}
    &\mathbb{E}^{\lambda^{a,*}_k,\lambda^{b,*}_k}_{k\delta}\big[U^a_{k+1}]\\&=\underset{\lambda^a_k\in\mathcal{U}_{[k\delta,(k+1)\delta]}}{\einf}\;\mathbb{E}^{\lambda^a_k,\lambda^{b,*}_k}_{k\delta}\big[U^a_{k+1}\big]\\
    &=\underset{\lambda^a_k\in\mathcal{U}_{[k\delta,(k+1)\delta]}}{\einf}\;\mathbb{E}^{\lambda^{a}_k,\lambda^{b,*}_{k}}_{k\delta}\Big[\underset{\tau^a\in\mathcal{T}^d_{(k+1)\delta,T},\lambda^a_{k+1}\in\mathcal{U}_{[(k+1)\delta,T+h]}}{\einf}\mathbb{E}^{\lambda^{a}_{k+1},\tilde\lambda^{b,*}_{k+1}}_{(k+1)\delta}\big[\frac{L^a_\tau+\mathcal{C}^a_{auc}+\frac{N^a_{\tau,\tau+h}(N^a_{\tau,\tau+h}-N^b_{\tau,\tau+h})}{K}}{\tau+h}\big]\Big].
\end{split}
\end{equation}
We aim at showing that
\begin{equation}\label{eq::int}
\begin{split}
    &\underset{\lambda^a_k\in\mathcal{U}_{[k\delta,(k+1)\delta]}}{\einf}\;\mathbb{E}^{\lambda^{a}_k,\lambda^{b,*}_{k}}_{k\delta}\Big[\underset{\tau^a\in\mathcal{T}_{(k+1)\delta,T},\lambda^a_{k+1}\in\mathcal{U}_{[(k+1)\delta,T+h]}}{\einf}\mathbb{E}^{\lambda^{a}_{k+1},\tilde\lambda^{b,*}_{k+1}}_{(k+1)\delta}\big[\frac{L^a_\tau+\mathcal{C}^a_{auc}+\frac{N^a_{\tau,\tau+h}(N^a_{\tau,\tau+h}-N^b_{\tau,\tau+h})}{K}}{\tau+h}\big]\Big]\\
    &=\underset{\tau^a\in\mathcal{T}_{(k+1)\delta,T},\lambda^a_k\in\mathcal{U}_{[k\delta,T+h]}}{\einf}\mathbb{E}^{\lambda^{a}_k,\tilde\lambda^{b,*}_k}_{k\delta}\big[\frac{L^a_\tau+\mathcal{C}^a_{auc}+\frac{N^a_{\tau,\tau+h}(N^a_{\tau,\tau+h}-N^b_{\tau,\tau+h})}{K}}{\tau+h}\big],
\end{split}
\end{equation}
where abusing notation slightly we write $\tau={T\wedge\tau^a\wedge\tau_{k+1}^*}$. Let $$ \chi^a_{T\wedge\tau^a\wedge\tau_{k+1}^*}=\frac{L^a_\tau+\mathcal{C}^a_{auc}+\frac{N^a_{\tau,\tau+h}(N^a_{\tau,\tau+h}-N^b_{\tau,\tau+h})}{K}}{\tau+h}$$ and $\tau^a$ be in $\mathcal{T}^d_{(k+1)\delta,T}$. Using the tower property we get

\begin{equation*}
\begin{split}
    \mathbb{E}^{\lambda^{a}_k,\tilde\lambda^{b,*}_k}_{k\delta}\big[\chi^a_{T\wedge\tau^a\wedge\tau_{k+1}^*}\big] =& \mathbb{E}^{\lambda^{a}_k,\tilde\lambda^{b,*}_k}_{k\delta}\Big[\mathbb{E}^{\lambda^{a}_k,\tilde\lambda^{b,*}_k}_{(k+1)\delta}\big[\chi^a_{T\wedge\tau^a\wedge\tau_{k+1}^*}\big]\Big]\\
    =& \mathbb{E}^{\lambda^{a}_k,\tilde\lambda^{b,*}_k}_{k\delta}\Big[\mathbb{E}^{\lambda^{a}_{k,[(k+1)\delta,T]},\tilde\lambda^{b,*}_{k+1}}_{(k+1)\delta}\big[\chi^a_{T\wedge\tau^a\wedge\tau_{k+1}^*}\big]\Big]\\
    =& \mathbb{E}^{\lambda^{a}_{k,[k\delta,(k+1)\delta]},\lambda^{b,*}_{k}}_{k\delta}\Big[\mathbb{E}^{\lambda^{a}_{k,[(k+1)\delta,T]},\tilde\lambda^{b,*}_{k+1}}_{(k+1)\delta}\big[\chi^a_{T\wedge\tau^a\wedge\tau_{k+1}^*}\big]\Big]\\
    \geq& \mathbb{E}^{\lambda^{a}_{k,[k\delta,(k+1)\delta]},\lambda^{b,*}_{k}}_{k\delta}\Big[\underset{\tau^a\in\mathcal{T}^d_{(k+1)\delta,T},\lambda^a_{k+1}\in\mathcal{U}_{[(k+1)\delta,T+h]}}{\einf}\mathbb{E}^{\lambda^{a}_{k+1},\tilde\lambda^{b,*}_{k+1}}_{(k+1)\delta}\big[\chi^a_{T\wedge\tau^a\wedge\tau_{k+1}^*}\big]\Big]\\
    \geq& \underset{\lambda^a_k\in\mathcal{U}^d_{[k\delta,(k+1)\delta]}}{\einf}\;\mathbb{E}^{\lambda^{a}_k,\lambda^{b,*}_{k}}_{k\delta}\Big[\underset{\tau^a\in\mathcal{T}^d_{(k+1)\delta,T},\lambda^a_{k+1}\in\mathcal{U}_{[(k+1)\delta,T+h]}}{\einf}\mathbb{E}^{\lambda^{a}_{k+1},\tilde\lambda^{b,*}_{k+1}}_{(k+1)\delta}\big[\chi^a_{T\wedge\tau^a\wedge\tau_{k+1}^*}\big]\Big],
\end{split}
\end{equation*}
which gives one inequality in \eqref{eq::int} by taking the essential infimum over $\tau^a\in\mathcal{T}^d_{(k+1)\delta,T}$ and $\lambda^a_k\in\mathcal{U}_{[k\delta,T+h]}$. We now turn to the other inequality. Let $\lambda^a\in\mathcal{U}_{[k\delta,(k+1)\delta]}$ and $\mu^a\in\mathcal{U}_{[(k+1)\delta,T+h]}$. We recall the definition $(\lambda^a\otimes_{(k+1)\delta}\mu^a)_u := \lambda^a_u\mathbf{1}_{u\leq(k+1)\delta}+\mu^a_u\mathbf{1}_{(k+1)\delta<u}$ so that $\lambda^a\otimes_{(k+1)\delta}\mu^a\in\mathcal{U}_{[k\delta,T+h]}$. Let $\tau^a\in\mathcal{T}^d_{(k+1)\delta,T}$. We have
\begin{equation}
\begin{split}\label{eq::ineq_disc_almost}
    \underset{\tau^a\in\mathcal{T}^d_{(k+1)\delta,T},\lambda^a_k\in\mathcal{U}_{[k\delta,T+h]}}{\einf}\mathbb{E}^{\lambda^{a}_{k},\tilde\lambda^{b,*}_k}_{k\delta}\big[\chi^a_{T\wedge\tau^a\wedge\tau_{k+1}^*}\big]&\leq \mathbb{E}^{\lambda^a\otimes_{(k+1)\delta}\mu^a,\tilde\lambda^{b,*}_k}_{k\delta}\big[\chi^a_{T\wedge\tau^a\wedge\tau_{k+1}^*}\big]\\
    &=\mathbb{E}^{\lambda^{a}_{[k\delta,(k+1)\delta]},\tilde\lambda^{b,*}_{k}}_{k\delta}\Big[\mathbb{E}^{\mu^a,\tilde\lambda^{b,*}_{k+1}}_{(k+1)\delta}\big[\chi^a_{T\wedge\tau^a\wedge\tau_{k+1}^*}\big]\Big].
\end{split}
\end{equation}
This inequality holds for any $\tau^a\in\mathcal{T}^d_{(k+1)\delta,T}$ and $\mu^a\in\mathcal{U}_{[(k+1)\delta,T+h]}$ and in particular for $\tau^a$ and $\mu^a$ such that
\begin{equation}\label{eq::inf_atteint}
    \mathbb{E}^{\mu^{a},\tilde\lambda^{b,*}_{k+1}}_{(k+1)\delta}\big[\chi^a_{T\wedge\tau^a\wedge\tau_{k+1}^*}\big] = \underset{\tau^a\in\mathcal{T}^d_{(k+1)\delta,T},\lambda^a_{k+1}\in\mathcal{U}_{[k\delta,(k+1)\delta]}}{\einf}\mathbb{E}^{\lambda^{a}_{k+1},\tilde\lambda^{b,*}_{k+1}}_{(k+1)\delta}\big[\chi^a_{T\wedge\tau^a\wedge\tau_{k+1}^*}\big].
\end{equation}
From the induction hypothesis, we know that such quantities exist (we can take $\tau^*_{k+1}$ and $\Tilde{\lambda}^{a,*}_{k+1}$). Thus from \eqref{eq::ineq_disc_almost} and \eqref{eq::inf_atteint}   
we have 
\begin{equation*}
\begin{split}
    &\underset{\tau^a\in\mathcal{T}^d_{(k+1)\delta,T},\lambda^a_k\in\mathcal{U}_{[k\delta,T+h]}}{\einf}\mathbb{E}^{\lambda^{a}_k,\tilde\lambda^{b,*}_k}_{k\delta}\big[\chi^a_{T\wedge\tau^a\wedge\tau_{k+1}^*}\big]\\
    &\leq\mathbb{E}^{\lambda^{a}_{k,[k\delta,(k+1)\delta]},\tilde\lambda^{b,*}_{k}}_{k\delta}\Big[\underset{\tau^a\in\mathcal{T}^d_{(k+1)\delta,T},\lambda^a_{k+1}\in\mathcal{U}_{[(k+1)\delta,T+h]}}{\einf}\mathbb{E}^{\lambda^{a}_{k+1},\tilde\lambda^{b,*}_{k+1}}_{(k+1)\delta}\big[\chi^a_{T\wedge\tau^a\wedge\tau_{k+1}^*}\big]\Big]\\
    &=\mathbb{E}^{\lambda^{a}_{k},\lambda^{b,*}_{k}}_{k\delta}\Big[\underset{\tau^a\in\mathcal{T}^d_{(k+1)\delta,T},\lambda^a_{k+1}\in\mathcal{U}_{[(k+1)\delta,T+h]}}{\einf}\mathbb{E}^{\lambda^{a}_{k+1},\tilde\lambda^{b,*}_{k+1}}_{(k+1)\delta}\big[\chi^a_{T\wedge\tau^a\wedge\tau_{k+1}^*}\big]\Big]
\end{split}
\end{equation*}
and we deduce the second inequality in \eqref{eq::int} by taking the essential infimum over $\lambda^a_k\in\mathcal{U}^d_k$.
From \eqref{eq::reecriture_U} and \eqref{eq::int} we deduce that
\begin{equation}\label{eq:recurrence}
    \mathbb{E}^{\lambda^{a,*}_k,\lambda^{b,*}_k}_{k\delta}\big[U^a_{k+1}] = \underset{\tau^a\in\mathcal{T}^d_{(k+1)\delta,T},\lambda^a_k\in\mathcal{U}_{[k\delta,T+h]}}{\einf}\mathbb{E}^{\lambda^{a}_k,\lambda^{b,*}_k}_{k\delta}\big[\frac{L^a_\tau+\mathcal{C}^a_{auc}+\frac{N^a_{\tau,\tau+h}(N^a_{\tau,\tau+h}-N^b_{\tau,\tau+h})}{K}}{\tau+h}\big].
\end{equation}
We conclude from \eqref{eq:recurrence} that
\begin{equation*}
\begin{split}
    &\underset{\tau^a\in\mathcal{T}^d_{k\delta,T}, \lambda^a\in\mathcal{U}_{[k\delta,T+h]}}{\einf}\mathbb{E}^{\lambda^a,\Tilde{\lambda}^{b,*}_{k}}_{k\delta}\big[\frac{L^a_\tau+\mathcal{C}^a_{auc}+\frac{N^a_{\tau,\tau+h}(N^a_{\tau,\tau+h}-N^b_{\tau,\tau+h})}{K}}{\tau+h}\big]\\
    &=\min\Big(\mathbb{E}^{\lambda^{a,*}_k,\lambda^{b,*}_k}_{k\delta}\big[U^a_{k+1}] ,\frac{L^a_{k\delta}+g^{\text{first}}_a({k\delta},,N^a_{k\delta},N^b_{k\delta})}{k\delta+h}\Big)
\end{split}
\end{equation*}
where $\tau = \tau^a\wedge\tau_{k+1}^*$ and symmetrically 
\begin{equation*}
\begin{split}
    &\underset{\tau^b\in\mathcal{T}^d_{k\delta,T}, \lambda^b\in\mathcal{U}_{[k\delta,T+h]}}{\esup}\mathbb{E}^{\Tilde{\lambda}^{a,*}_{k},\lambda^b}_{k\delta}\big[\frac{-L^b_\tau-\mathcal{C}^b_{auc}+\frac{N^b_{\tau,\tau+h}(N^a_{\tau,\tau+h}-N^b_{\tau,\tau+h})}{K}}{\tau+h}]\\
    &=\max\Big(\mathbb{E}^{\lambda^{a,*}_k,\lambda^{b,*}_k}_{k\delta}\big[U^b_{k+1}],\frac{-L^b_{k\delta}+g^{\text{first}}_b({k\delta},N^a_{k\delta},N^b_{k\delta})}{k\delta+h}\Big)
\end{split}
\end{equation*}
where $\tau = \tau_{k+1}^*\wedge\tau^b$.
So, if either
\begin{equation*}
    \mathbb{E}^{\Tilde{\lambda}^{a,*}_{k},\lambda^b}_{k\delta}\big[U^a_{k+1}\big]<\frac{L^a_{k\delta}+g^{\text{first}}_a({k\delta},N^a_{k\delta},N^b_{k\delta})}{k\delta+h}
\end{equation*}
or
\begin{equation*}
    \mathbb{E}^{\Tilde{\lambda}^{a,*}_{k},\lambda^b}_{k\delta}\big[U^b_{k+1}\big]>\frac{-L^b_{k\delta}+g^{\text{first}}_b({k\delta},N^a_{k\delta},N^b_{k\delta})}{k\delta+h},
\end{equation*}
the extrema are not reached at $k\delta$, \textit{i.e.} the optimal stopping times are still equal to $\tau^*_{k+1}$ and the trading rates are given by $\lambda^{a,*}_k\otimes_{(k+1)\delta}\lambda^{a,*}_{k+1} = \Tilde{\lambda}^{a,*}_k$ and $\lambda^{b,*}_k\otimes_{(k+1)\delta}\lambda^{b,*}_{k+1} = \Tilde{\lambda}^{b,*}_k$. Else at least one player triggers an auction and the optimal stopping times are equal to $k\delta$. Consequently,

\[  \tau^{*}_k=\delta\inf\Big\{l\in\llbracket k,T/\delta\rrbracket\text{, }U^a_l =\frac{L^a_{l\delta}+g^{\text{first}}_a({l\delta},N^a_{l\delta},N^b_{l\delta})}{l\delta+h}\text{ or } U^b_l=\frac{-L^b_{l\delta}+g^{\text{first}}_b({l\delta},N^a_{l\delta},N^b_{l\delta})}{l\delta+h}\Big\}.
\]

\subsection{Proof of Theorem \ref{thm:purediscrete}}\label{app:purediscrete}
\textit{Step 1: Construction of a (good) stopping strategy.}
Fix $\tau^b\in\mathcal{T}_{0,T}$ and $\lambda^b\in\mathcal{U}$. Let $\varepsilon>0$, $\lambda^*\in\mathcal{U}$ and  $\tau^*\in\mathcal{T}_{0,T}$ such that
\begin{equation*}
    \mathbb{E}^{\lambda^*,\lambda^b} \big[\frac{L^a_{\tau^*\wedge \tau^b}+\xi^a_{\tau^*\wedge\tau^b}}{\tau^*\wedge\tau^b+h}\big] \leq  \underset{\tau\in\mathcal{T}_{0,T},\lambda\in\mathcal{U}}{\inf}\mathbb{E}^{\lambda,\lambda^b} \big[\frac{L^a_{\tau\wedge\tau^b}+\xi^a_{\tau\wedge\tau^b}}{\tau\wedge \tau^b+h}\big]+\varepsilon
\end{equation*}
and define $\tau^*_d = \delta\lceil\frac{\tau^*}{\delta}\rceil\mathbf{1}_{\tau^*<\tau^b}+\mathbf{1}_{\tau^*\geq\tau^b}\tau^*$. Also let $\lambda^a = \lambda^*\otimes_{\tau^*}\hat\lambda_-$ where $\hat\lambda_-\in\mathcal{U}$ is the constant strategy equal to $\lambda_-$.\\

\textit{Step 2: Comparison with the optimal payoff.}
We decompose the error made by choosing $\tau^*_d$ instead of $\tau^*$ into three terms:
\begin{align*}
    0&\leq \mathbb{E}^{\lambda^a,\lambda^b} \big[\frac{L^a_{\tau^*_d\wedge\tau^b}+\xi^a_{\tau^*_d\wedge\tau^b}}{\tau^*_d\wedge \tau^b+h}\big]-\underset{\tau\in\mathcal{T}_{0,T},\lambda\in\mathcal{U}}{\inf}\mathbb{E}^{\lambda,\lambda^b} \big[\frac{L^a_{\tau\wedge\tau^b}+\xi^a_{\tau\wedge\tau^b}}{\tau\wedge \tau^b+h}\big]\\
    &\leq \mathbb{E}^{\lambda^a,\lambda^b} \big[\frac{L^a_{\tau^*_d\wedge\tau^b}+\xi^a_{\tau^*_d\wedge\tau^b}}{\tau^*_d\wedge \tau^b+h}\big]-\mathbb{E}^{\lambda^*,\lambda^b} \big[\frac{L^a_{\tau^*\wedge\tau^b}+\xi^a_{\tau^*\wedge\tau^b}}{\tau^*\wedge \tau^b+h}\big]+\varepsilon\\
    & = \mathbb{E}^{\lambda^a,\lambda^b} \big[\frac{L^a_{\tau^*_d\wedge\tau^b}+\xi^a_{\tau^*_d\wedge\tau^b}}{\tau^*_d\wedge \tau^b+h}-\frac{L^a_{\tau^*\wedge\tau^b}+\xi^a_{\tau^*\wedge\tau^b}}{\tau^*\wedge \tau^b+h}\big]+\varepsilon\\
    & = \mathbb{E}^{\lambda^a,\lambda^b} \big[\mathbf{1}_{\tau^*<\tau^b}(\frac{L^a_{\tau^*_d\wedge\tau^b}+\xi^a_{\tau^*_d\wedge\tau^b}}{\tau^*_d\wedge \tau^b+h}-\frac{L^a_{\tau^*\wedge\tau^b}+\xi^a_{\tau^*\wedge\tau^b}}{\tau^*\wedge \tau^b+h})\big]+\varepsilon\\
    & = \mathbb{E}^{\lambda^a,\lambda^b} \big[\mathbf{1}_{\tau^*<\tau^b}(\frac{L^a_{\tau^*_d\wedge\tau^b}+\xi^a_{\tau^*_d\wedge\tau^b}}{\tau^*_d\wedge \tau^b+h}-\frac{L^a_{\tau^*\wedge\tau^b}+g_a^{\text{first}}({\tau^*\wedge\tau^b})}{\tau^*\wedge \tau^b+h})\big]+\varepsilon\\
    & = \mathbb{E}^{\lambda^a,\lambda^b} \big[\mathbf{1}_{\tau^*<\tau^b}(\frac{L^a_{\tau^*_d\wedge\tau^b}+g_a^{\text{first}}({\tau^*_d\wedge\tau^b},N^a_{\tau^*_d\wedge\tau^b},N^b_{\tau^*_d\wedge\tau^b})}{\tau^*_d\wedge \tau^b+h}-\frac{L^a_{\tau^*\wedge\tau^b}+g_a^{\text{first}}({\tau^*\wedge\tau^b},N^a_{\tau^*\wedge\tau^b},N^b_{\tau^*\wedge\tau^b})}{\tau^*\wedge \tau^b+h})\big]\\
    & + \mathbb{E}^{\lambda^a,\lambda^b} \big[\mathbf{1}_{\tau^*<\tau^b}(\frac{L^a_{\tau^*_d\wedge\tau^b}+\xi^a_{\tau^*_d\wedge\tau^b}}{\tau^*_d\wedge \tau^b+h}-\frac{L^a_{\tau^*_d\wedge\tau^b}+g_a^{\text{first}}({\tau^*_d\wedge\tau^b},N^a_{\tau^*_d\wedge\tau^b},N^b_{\tau^*_d\wedge\tau^b})}{\tau^*_d\wedge \tau^b+h})\big]+\varepsilon.
\end{align*}

We have $\tau^{*,d}-\tau^*\underset{\delta\rightarrow 0}{\rightarrow} 0$ a.s. Furthermore, since $v^a\tau$ and $v^b\tau$ are replaced by the nearest integer and using that $\frac{1}{2v^a\delta}\in\mathbb{N}$ and $\frac{1}{2v^b\delta}\in \mathbb N$, we get
$$g_a^{\text{first}}({\tau^*_d\wedge\tau^b},N^a_{\tau^*_d\wedge\tau^b},N^b_{\tau^*_d\wedge\tau^b})=g_a^{\text{first}}({\tau^*\wedge\tau^b},N^a_{\tau^*_d\wedge\tau^b},N^b_{\tau^*_d\wedge\tau^b}).$$
So, using that the intensities of the Poisson processes are bounded, we have
$$\mathbb{E}^{\lambda^a,\lambda^b}\big[\mathbf{1}_{\tau^*<\tau^b}(\frac{L^a_{\tau^*_d\wedge\tau^b}+g_a^{\text{first}}({\tau^*_d\wedge\tau^b},N^a_{\tau^*_d\wedge\tau^b},N^b_{\tau^*_d\wedge\tau^b})}{\tau^*_d\wedge \tau^b+h}-\frac{L^a_{\tau^*\wedge\tau^b}+g_a^{\text{first}}({\tau^*\wedge\tau^b},N^a_{\tau^*\wedge\tau^b},N^b_{\tau^*\wedge\tau^b})}{\tau^*\wedge \tau^b+h})\big]\underset{\delta\rightarrow 0}{\rightarrow} 0,$$
uniformly in $\tau^*$, $\tau^b$, $\lambda^a$ and $\lambda^b$.
Also,
\begin{align*}
    &\mathbb{E}^{\lambda^a,\lambda^b} \big[\mathbf{1}_{\tau^*<\tau^b}(\frac{L^a_{\tau^*_d\wedge\tau^b}+\xi^a_{\tau^*_d\wedge\tau^b}}{\tau^*_d\wedge \tau^b+h}-\frac{L^a_{\tau^*_d\wedge\tau^b}+g_a^{\text{first}}({\tau^*_d\wedge\tau^b},N^a_{\tau^*_d\wedge\tau^b},N^b_{\tau^*_d\wedge\tau^b})}{\tau^*_d\wedge \tau^b+h})\big]\\
    &\leq \frac{1}{h}\mathbb{E}^{\lambda^a,\lambda^b} \big[\max(0,\xi^a_{\tau^*_d\wedge\tau^b}-g_a^{\text{first}}({\tau^*_d\wedge\tau^b},N^a_{\tau^*_d\wedge\tau^b},N^b_{\tau^*_d\wedge\tau^b}))\big]\\
    &\leq \frac{1}{h}\underset{\lambda^a_0\in\mathcal{U},\lambda^b_0\in\mathcal{U}}{\sup}\mathbb{E}^{\lambda^a_0,\lambda^b_0} \big[\underset{[0,T]}{\sup}(\max(\max(g_a^{\text{sim}}(t,N^a_t,N^b_t),g_a^{\text{second}}(t,N^a_t,N^b_t))-g_a^{\text{first}}(t,N^a_t,N^b_t),0))\big]\\
    &=\varepsilon_a.
\end{align*}
\textit{Step 3: Conclusion.}
We have shown that for any $\varepsilon>0$, there exists $\hat \delta^a>0$ such that, if $\delta\leq\hat\delta^a$, then for any $\tau^b$ and $\lambda^b$ chosen by Player $b$, we can find some $\tau^*_d\in\mathcal{T}_{0,T}$ and $\lambda^a\in\mathcal{U}$ such that
$$    0\leq \mathbb{E}^{\lambda^a,\lambda^b} \big[\frac{L^a_{\tau^*_d\wedge\tau^b}+\xi^a_{\tau*_d\wedge\tau^b}}{\tau^*_d\wedge \tau^b+h}\big]-\underset{\tau\in\mathcal{T}_{0,T},\lambda\in\mathcal{U}}{\inf}\mathbb{E}^{\lambda,\lambda^b} \big[\frac{L^a_{\tau\wedge\tau^b}+\xi^a_{\tau\wedge\tau^b}}{\tau\wedge \tau^b+h}\big]\leq 2\varepsilon+ \varepsilon^a.
$$Now let $((\lambda^a,\lambda^b),(\tau^a,\tau^b))$ be an OLNED for some $\delta<\hat\delta^a$ and let $\tau^*_d\in\mathcal T_{0,T}$ and $\lambda^a\in\mathcal{U}$ as in \textit{Step 1}. Remark that $\tau^*_d\in\mathcal{T}^d_{0,T}$ and so
\begin{align*}
    0&\leq \mathbb{E}^{\lambda^a,\lambda^b} \big[\frac{L^a_{\tau^a\wedge\tau^b}+\xi^a_{\tau^a\wedge\tau^b}}{\tau^a\wedge \tau^b+h}\big]-\underset{\tau\in\mathcal{T}_{0,T},\lambda\in\mathcal{U}}{\inf}\mathbb{E}^{\lambda,\lambda^b} \big[\frac{L^a_{\tau\wedge\tau^b}+\xi^a_{\tau\wedge\tau^b}}{\tau\wedge \tau^b+h}\big]\\
    &=\underset{\tau\in\mathcal{T}^d_{0,T},\lambda\in\mathcal{U}}{\inf}\mathbb{E}^{\lambda,\lambda^b} \big[\frac{L^a_{\tau\wedge\tau^b}+\xi^a_{\tau\wedge\tau^b}}{\tau\wedge \tau^b+h}\big]-\underset{\tau\in\mathcal{T}_{0,T},\lambda\in\mathcal{U}}{\inf}\mathbb{E}^{\lambda,\lambda^b} \big[\frac{L^a_{\tau\wedge\tau^b}+\xi^a_{\tau\wedge\tau^b}}{\tau\wedge \tau^b+h}\big]\\
    &\leq \mathbb{E}^{\lambda^a,\lambda^b} \big[\frac{L^a_{\tau^*_d\wedge\tau^b}+\xi^a_{\tau*_d\wedge\tau^b}}{\tau^*_d\wedge \tau^b+h}\big]-\underset{\tau\in\mathcal{T}_{0,T},\lambda\in\mathcal{U}}{\inf}\mathbb{E}^{\lambda,\lambda^b} \big[\frac{L^a_{\tau\wedge\tau^b}+\xi^a_{\tau\wedge\tau^b}}{\tau\wedge \tau^b+h}\big]\\
    &\leq 2\varepsilon+ \varepsilon^a.
\end{align*}
Similarly, we find $\hat \delta^b>0$ such that if $((\lambda^a,\lambda^b),(\tau^a,\tau^b))$ is an OLNED for some $\delta<\hat\delta^b$, then
$$
    0\leq \underset{\tau\in\mathcal{T}_{0,T},\lambda\in\mathcal{U}}{\sup}\mathbb{E}^{\lambda^a,\lambda} \big[\frac{-L^b_{\tau^a\wedge\tau}+\xi^b_{\tau^a\wedge\tau}}{\tau^a\wedge \tau+h}\big]-\mathbb{E}^{\lambda^a,\lambda^b} \big[\frac{-L^b_{\tau^a\wedge\tau^b}+\xi^b_{\tau^a\wedge\tau^b}}{\tau^a\wedge \tau^b+h}\big]\leq 2\varepsilon+ \varepsilon^b.$$
\section{Algorithms}\label{algo}

\subsection{Value functions when $\hat n=\hat n_{ab}=0$}
\begin{algorithm}[H]\label{algovaluediscrete}
\SetAlgoLined
\KwResult{Value functions and probability of triggering an auction}
 Set 
 \begin{equation*}
    (U^a_{T/\delta},U^b_{T/\delta}) = (\frac{L^a_T+g^{\text{T}}_a(T,P^*_T,N^a_T,N^b_T)}{T+h}, \frac{-L^b_T+g^{\text{T}}_b(T,P^*_T,N^a_T,N^b_T)}{T+h})
\end{equation*}
 \For{$k\in\{T/\delta-1,...,0\}$}{
  Let
$(\lambda^{a,*}_{k},\lambda^b_{k})\in\mathcal{U}_{k}^2$ such that
\begin{equation*}
\begin{cases}
  \mathbb{E}^{\lambda^{a,*}_k,\lambda^{b,*}_k}_{k\delta}\big[U^a_{k+1}\big] = \underset{\lambda^a\in\mathcal{U}_{{[k\delta,(k+1)\delta]}}}{\einf}\;\mathbb{E}^{\lambda^a,\lambda^{b,*}_k}_{k\delta}\big[U^a_{k+1}\big]\\
\mathbb{E}^{\lambda^{a,*}_k,\lambda^{b,*}_k}_{k\delta}\big[U^b_{k+1}\big] = \underset{\lambda^b\in\mathcal{U}_{{[k\delta,(k+1)\delta]}}}{\esup}\;\mathbb{E}^{\lambda^{a,*}_k,\lambda^b}_{k\delta}\big[U^b_{k+1}\big],\\
\end{cases}
\end{equation*}
\eIf{$\mathbb{E}^{\lambda^{a,*}_{k},\lambda^{b,*}_{k}}_{k\delta}\big[U^a_{k+1}\big]\leq\frac{L^a_{k\delta}+g^{\text{first}}_a({k\delta},P^*_{k\delta},N^a_{k\delta},N^b_{k\delta})}{k\delta+h}$ and $ \mathbb{E}^{\lambda^{a,*}_{k},\lambda^{b,*}_{k}}_{k\delta}\big[U^b_{k+1}\big]\geq\frac{-L^b_{k\delta}+g^{\text{first}}_b({k\delta},P^*_{k\delta},N^a_{k\delta},N^b_{k\delta})}{k\delta+h}$}{
Set
$(U^a_k,U^b_k) = (\mathbb{E}^{\lambda^{a,*}_{k},\lambda^{b,*}_{k}}_{k\delta}\big[U^a_{k+1}\big],\mathbb{E}^{\lambda^{a,*}_{k},\lambda^{b,*}_{k}}_{k\delta}\big[U^b_{k+1}\big])$
}{Set
$(U^a_k,U^b_k) = (\frac{L^a_{k\delta}+g^{\text{first}}_a({k\delta},P^*_{k\delta},N^a_{k\delta},N^b_{k\delta})}{k\delta+h}, \frac{-L^b_{k\delta}+g^{\text{first}}_b({k\delta},P^*_{k\delta},N^a_{k\delta},N^b_{k\delta})}{k\delta+h})$}

 }
 \caption{Computation of the value functions in the discretised game}
\end{algorithm}

\subsection{Value functions in the general case: randomised discrete stopping time}
\begin{algorithm}[H]\label{algovaluediscretegeneral}
\SetAlgoLined
\KwResult{Value functions and probability of triggering an auction}
 Set 
 \begin{equation*}
    (U^a_{T/\delta},U^b_{T/\delta}) = (\frac{L^a_T+g^{\text{T}}_a(T,P^*_T,N^a_T,N^b_T)}{T+h}, \frac{-L^b_T+g^{\text{T}}_b(T,P^*_T,N^a_T,N^b_T)}{T+h})
\end{equation*}
 \For{$k\in\{T/\delta-1,...,0\}$}{
  Let
\begin{itemize}
    \item $(\lambda^{a,*}_{k},\lambda^{b,*}_{k})\in\mathcal{U}_{k}^2$ such that
\begin{equation*}
\begin{cases}
  \mathbb{E}^{\lambda^{a,*}_k,\lambda^{b,*}_k}_{k\delta}\big[U^a_{k+1}\big] = \underset{\lambda^a\in\mathcal{U}_{[k\delta,(k+1)\delta]}}{\einf}\;\mathbb{E}^{\lambda^a,\lambda^{b,*}_k}_{k\delta}\big[U^a_{k+1}\big]\\
\mathbb{E}^{\lambda^{a,*}_k,\lambda^{b,*}_k}_{k\delta}\big[U^b_{k+1}\big] = \underset{\lambda^b\in\mathcal{U}_{[k\delta,(k+1)\delta]}}{\esup}\;\mathbb{E}^{\lambda^{a,*}_k,\lambda^b}_{k\delta}\big[U^b_{k+1}\big],\\
\end{cases}
\end{equation*}
    \item $(p^a_k,p^b_k)\in[0,1]^2$ such that $(p^a_k,1-p^a_k)$ and $(p^b_k,1-p^b_k)$ define the mixed strategies\\
    of a Nash equilibrium for the discrete game of Table \ref{tab::mixed}.
    \end{itemize}
Set
    \begin{equation*}
    \begin{cases}
        \resizebox{\hsize}{!}{$U^a_k = p^a_k p^b_k \frac{L^a_{k\delta}+g^{\text{sim}}_a}{k\delta+h} + p^a_k(1-p^b_k)\frac{L^a_{k\delta}+g^{\text{first}}_a}{k\delta+h} + (1-p^a_k)p^b_k\frac{L^a_{k\delta}+g^{\text{second}}_a}{k\delta+h}+ (1-p^a_k)(1-p^b_k)\mathbb{E}^{\lambda^{a,*}_k,\lambda^{b,*}_k}_{k\delta}\big[U^a_{k+1}\big]$}\\[0.3em]
        \resizebox{\hsize}{!}{$U^b_k = p^a_k p^b_k \frac{-L^b_{k\delta}+g^{\text{sim}}_b}{k\delta+h} + p^a_k(1-p^b_k)\frac{-L^b_{k\delta}+g^{\text{second}}_b}{k\delta+h} + (1-p^a_k)p^b_k\frac{-L^b_{k\delta}+g^{\text{first}}_b}{k\delta+h}+ (1-p^a_k)(1-p^b_k)\mathbb{E}^{\lambda^{a,*}_k,\lambda^{b,*}_k}_{k\delta}\big[U^b_{k+1}\big]$}
    \end{cases}
    \end{equation*}
    where the arguments of the functions $g^{\text{first}}_a$, $g^{\text{second}}_a$, $g^{\text{sim}}_a$, $g^{\text{first}}_b$, $g^{\text{second}}_b$ and $g^{\text{sim}}_b$ are $(k\delta,P^*_{k\delta},N^a_{k\delta},N^b_{k\delta})$.
  
 }
 \caption{Computation of the value functions in the discretised game}
\end{algorithm}

\subsection{Average duration of the continuous trading phase}\label{section:algo_duration}

The average duration is defined by $$E_k = \mathbb{E}^{\lambda^{a,*}_k,\lambda^{b,*}_k}_{k\delta}\big[\tau-k\delta| \tau\geq k\delta\big]$$ for $k\in\{0,...,T/\delta\}$. In particular $E_0=\mathbb{E}^{\lambda^{a,*},\lambda^{b,*}}[\tau]$. Assuming that it is a Markovian function of the state variables $P^*,N^a,N^b,L^a,L^b$ and that the PDE obtained from the Feyman-Kac formula has a unique solution, we compute $E_k$ with the following algorithm. 

\begin{algorithm}[H]\label{algo:meantime}
\SetAlgoLined
\KwResult{Average duration of the continuous phase}
 Set 
 \begin{equation*}
    E_{T/\delta} = 0
\end{equation*}
 \For{$k\in\{T/\delta-1,...,0\}$}{
  Compute $\mathbb{E}^{\lambda^{a,*}_k,\lambda^{b,*}_k}_{k\delta}\big[E_{k+1}\big]$ with the Feyman-Kac formula: $\mathbb{E}^{\lambda^{a,*}_k,\lambda^{b,*}_k}_{k\delta}\big[E_{k+1}\big] = e_{k\delta}$ where $e$ is the solution of the equation
\begin{equation*}
\begin{cases}
    &e_{(k+1)\delta} = E_{k+1}\\ 
    &\partial_t e+\frac{\sigma^2}{2}\partial^2_{pp}e +1+q(v^a t-n^a)^2\partial_{l^a}e+q(v^b t-n^b)^2\partial_{l^b}e+\mathbf{1}_{P>p}\lambda^{a,*,k}_t D^a e+\mathbf{1}_{P<p}\lambda^{b,*,k}_t D^b e = 0 \\
    &\text{  for }t\in[k\delta,(k+1)\delta)\\
\end{cases}
\end{equation*}
Set
    \begin{equation*}
    E_k = (1-p^a_k)(1-p^b_k)\mathbb{E}^{\lambda^{a,*}_k,\lambda^{b,*}_k}_{k\delta}\big[E_{k+1}\big]
    \end{equation*}
  
 }
 \caption{Computation of the average duration of the continuous phase in the discretised game}
\end{algorithm}
$~$\\
\noindent We recall that the operators $D^a,D^b$ are defined in Appendix \ref{app:verif} as the change in the value functions due to a trade between Player $a$ and the market maker, and Player $b$ and the market maker, respectively. Also, $\partial_{l^a}$ and $\partial_{l^b}$ are the spatial derivatives with respect to the cash processes $L^a$ and $L^b$ of Player $a$ and Player $b$.
\newpage
\section{Comparison \textit{ad-hoc} auctions, periodic auctions and CLOB for small penalty}
\label{app:q005}
\begin{table}[h!]
\centering
    \begin{tabular}{|c || c c | c | c | c c | c | c||} 
         \hline
        &\multicolumn{4}{c|}{$V^a$ (1e-6)}&\multicolumn{4}{c||}{Average duration}\\
        \hline
        Market design & \multicolumn{2}{c|}{\shortstack{$h = 20$,\\ $\hat{n}=3$}}&\shortstack{$h = 20$,\\ $\hat{n}=1$}&CLOB& \multicolumn{2}{c|}{\shortstack{$h = 20$,\\ $\hat{n}=3$}}&\shortstack{$h = 20$,\\ $\hat{n}=1$}&CLOB\\
        \hline
         continuous trading allowed&Yes&No&No&No&Yes&No&No&No  \\ [0.9ex] 
         \hline\hline
         $v^a=0.1$, $v^b=0.1$&16458.6&27424.3&12000.5&10000.0&40.9s&5.6s&0.0s&10.0s \\ 
         \hline
         $v^a=0.05$, $v^b=0.1$&6533.5&11581.1&2000.5&5000.0&56.9s&14.9s&0.0s&10.0s \\  
         \hline
         $v^a=0.1$, $v^b=0.05$&14196.4&22874.1&12000.5&10000.0&56.9s&14.9s&0.0s&20.0s \\  
         \hline
         $v^a=0.15$, $v^b=0.1$&27287.7&30051.5&31501.0&15000.0&27.1s&3.9s&0.0s&6.7s \\  
         \hline
         $v^a=0.1$, $v^b=0.15$&16450.0&16309.9&9510.7&10000.0&27.1s&3.9s&0.0s&10.0s \\  
         \hline
    \end{tabular}
    \caption{$V^a$ and average duration of the continuous trading phase for different values of $v^a$ and $v^b$ with $q=0.005$.}
    \label{tab::compa_q0005}
\end{table}

\end{appendix}

\end{document}